\documentclass[letterpaper,11pt]{article}
\usepackage[margin=1in]{geometry}
\usepackage{amsfonts}
\usepackage{amsmath}
\usepackage{amsthm}
\usepackage{amssymb}
\usepackage{mathtools}
\usepackage[plain]{algorithm}
\usepackage[noend]{algpseudocode}
\usepackage[hyphens]{url}
\usepackage{hyperref}
\usepackage{graphicx}
\usepackage{framed}
\usepackage{subfigure}
\usepackage{xspace}
\usepackage{units}
\usepackage{authblk}

\newtheorem{theorem}{Theorem}
\newtheorem{corollary}{Corollary}

\newtheorem{lemma}{Lemma}

\newtheorem{remark}{Remark}

\makeatletter\renewcommand{\ALG@name}{Algorithm}\makeatother

\begin{document}

\title{Meeting in a Polygon by Anonymous Oblivious Robots}

\author[1]{G.\,A. Di Luna}
       \author[2]{P. Flocchini}
        \author[3]{N. Santoro}
        \author[4]{G. Viglietta}
        \author[5]{M. Yamashita}
        
\affil[1]{{\footnotesize Aix-Marseille University, LIS Laboratory,
	163, avenue de Luminy, Marseille, France,
	\texttt{giuseppe.diluna@lif.univ-mrs.fr}}}
\affil[2]{{\footnotesize University of Ottawa, School of Electrical Engineering and Computer Science,
              800 King Edward Avenue, Ottawa, Ontario,
              \texttt{ paola.flocchini@uottawa.ca}}}
\affil[3]{{\footnotesize Carleton University, School of Computer Science,
            1125 Colonel By Drive, Ottawa, Ontario,
           \texttt{santoro@scs.carleton.ca}}}
\affil[4]{{\footnotesize JAIST, School of Information Science,
           1-1 Asahidai, Nomi, Ishikawa 923-1211, Japan,
           \texttt{johnny@jaist.ac.jp}}}
\affil[5]{{\footnotesize Kyushu University, Department of Computer Science and Communication Engineering, 
           Motooka, Fukuoka, 819-0395, Japan,
           \texttt{mark@inf.kyushu-u.ac.jp}}}

\date{}

\maketitle

\begin{abstract}
The \emph{Meeting problem} for $k\geq 2$ searchers in a polygon $P$ (possibly with holes) consists in making the searchers move within $P$, according to a distributed algorithm, in such a way that at least two of them eventually come to see each other, regardless of their initial positions. The polygon is initially unknown to the searchers, and its edges obstruct both movement and vision. Depending on the shape of $P$, we minimize the number of searchers $k$ for which the Meeting problem is solvable. Specifically, if $P$ has a rotational symmetry of order $\sigma$ (where $\sigma=1$ corresponds to no rotational symmetry), we prove that $k=\sigma+1$ searchers are sufficient, and the bound is tight. Furthermore, we give an improved algorithm that optimally solves the Meeting problem with $k=2$ searchers in all polygons whose barycenter is not in a hole (which includes the polygons with no holes). Our algorithms can be implemented in a variety of standard models of mobile robots operating in Look-Compute-Move cycles. For instance, if the searchers have memory but are anonymous, asynchronous, and have no agreement on a coordinate system or a notion of clockwise direction, then our algorithms work even if the initial memory contents of the searchers are arbitrary and possibly misleading. Moreover, oblivious searchers can execute our algorithms as well, encoding information by carefully positioning themselves within the polygon. This code is computable with basic arithmetic operations (provided that the coordinates of the polygon's vertices are algebraic real numbers in some global coordinate system), and each searcher can geometrically construct its own destination point at each cycle using only a compass and a straightedge. We stress that such memoryless searchers may be located anywhere in the polygon when the execution begins, and hence the information they initially encode is arbitrary. Our algorithms use a self-stabilizing map construction subroutine which is of independent interest.
\end{abstract}

\newpage

\section{Introduction}\label{s:1}

\subsection{Framework}\label{s:1.1}
\paragraph{Meeting problem.}
Consider a set of $k\geq 2$ autonomous mobile robots, modeled as geometric points located in a polygonal enclosure $P$, which may contain holes. The boundary of $P$ limits both visibility and mobility, in that robots cannot move or see through the edges of $P$. Each robot observes the visible portion of $P$ (taking an instantaneous snapshot of it), executes an algorithm to compute a visible destination point, and then moves to that point. Such a \emph{Look-Compute-Move cycle} is repeated forever by every robot, each time taking a new snapshot and moving to a newly computed point.\footnote{The typical assumption in this model is that a searcher's local reference frame retains its orientation, scale, and handedness after each move. We will make this assumption as well, although it is not strictly needed by our algorithms (see the footnote in Section~\ref{s:2}).} In this paper we study the \emph{Meeting problem}, which prescribes the $k$ robots to move in such a way that eventually at least two of them come to see each other and become ``mutually aware''. We will refer to these robots as \emph{$P$-searchers}, or simply \emph{searchers}.

\paragraph{Searchers' limitations.}
Our searchers are severely limited, which makes the Meeting problem harder to solve.
\begin{itemize}
\item They do not know the shape of $P$ in advance, nor their whereabouts within $P$.
\item They are anonymous, implying that they all execute the same algorithm to determine their destination points.
\item They are oblivious, meaning that each destination point is computed based only on the last snapshot taken, while older snapshots are forgotten, and no memory is retained between cycles.\footnote{In Section~\ref{s:3} we will drop this assumption in order to give a cleaner exposition of our algorithms. In Section~\ref{s:4} we will restore the assumption and show how to extend our algorithms to oblivious searchers.}
\item They are deterministic, meaning that they cannot resort to randomness in their computations. Hence, their task has to be accomplished in all cases, as opposed to ``with high probability''.
\item They are asynchronous, in the sense that we make no assumptions on how fast each searcher completes a Look-Compute-Move cycle compared to the others. These parameters are dynamic and are controlled by an adversarial \emph{scheduler}.
\item They are disoriented, which means that they have no magnetic compasses, GPS devices, or agreements of any kind. Each searcher has its own independent local orientation, unit of length, and handedness.
\item They are silent, in that they cannot communicate with one another in any way. In particular, there is no shared memory, and the information contained in a snapshot can only be accessed by the searcher who took that snapshot.
\item They have arbitrary initial locations within $P$.
\end{itemize}
The polygon $P$ is anonymous, as well. In particular, its vertices do not carry labels, and can only be distinguished by their relative positions. There are no other landmarks or objects that can help the searchers orient themselves. Our goal is to design an algorithm that allows these searchers to solve the Meeting problem regardless of their initial locations within $P$ and regardless of how the adversarial scheduler decides to control their behaviors (by slowing down some and speeding up others). We emphasize that a searcher's snapshot contains the full geometry of the visible part of $P$, in contrast with other models where only some geometric information is available (e.g., only angles and no distances, as in~\cite{angles}).

\paragraph{Applications.}
In real-life applications, being in line of sight may allow robots to communicate in environments where non-optical means of communication are unavailable or impractical. For instance, free-space optical communication~\cite{S2} is a technology that uses a laser to implement a communication channel that is resilient to jamming and radio-frequency noise.

Solving the Meeting problem is a necessary preliminary step to more complex tasks, such as space coverage~\cite{C1} or the extensively studied \emph{Gathering problem}, where all $k$ robots have to physically reach the same point and stop there. In the special case of $k=2$ robots, the Gathering problem is also called \emph{Rendezvous problem}. Clearly, the terminating condition of the Meeting problem is more relaxed than that of Gathering; hence, any solution to the Gathering problem would also solve Meeting. Unfortunately, no solution to the Gathering problem in the setting considered here exists in the literature (see Section~\ref{s:1.3}), and to the best of our knowledge there are no previous results on the Meeting problem.

In fact, given our searchers' many handicaps, and especially their lack of memory and orientation, it is hard to see how they could solve any non-trivial problem at all. Nonetheless, in this paper we will present the surprising result that the Meeting problem is solvable in almost every polygon, even for $k=2$ searchers.

\subsection{Our Contributions}\label{s:1.2}
\paragraph{Techniques.}
Since our searchers are disoriented and have no kind of \emph{a-priori} agreement, they must use the geometric features of $P$ to implicitly agree on some ``landmarks'' which can help them in their task. In order to identify such landmarks, each searcher has to visit $P$ and construct a map of it. But this cannot be done straightforwardly, because searchers are oblivious, and they forget everything as soon as they move. To cope with this handicap, they carefully move within $P$ in such a way as to implicitly encode information as their distance from the closest vertex.

This positional encoding technique poses some obvious difficulties. First, it greatly limits the freedom of the searchers: they have to do precise movements to encode the correct information, and still manage to visit all of $P$ and update the map as they go. Second, since searchers can be located anywhere in $P$ when the execution starts, they could be implicitly encoding anything. This includes misleading information, such as a false map of $P$ that happens to be locally coherent with the surroundings of the searcher. Therefore, a searcher can never rely on the information it is implicitly encoding, but it must constantly re-visit the entire polygon to make sure that the map it is encoding is correct.\footnote{For this reason, there is no distributed algorithm that, for every polygon $P$, allows a team of memoryless searchers to either solve the Meeting problem in $P$ or terminate if the problem is unsolvable in $P$.}

Hence, searchers cannot simply agree on a landmark and sit on it waiting for one another, because that would prevent them from re-visiting $P$. This inconvenience drastically complicates the Meeting problem, and forces the searchers to follow relatively complicated movement patterns that make at least two of them necessarily meet.

There is also a subtle problem with the actual encoding of complex data as the distance from a point, which is a single real number. One could naively pack several real numbers into one by interleaving their digits, but this encoding would not be computable by real random-access machines~\cite{blum}. Hence, we propose a more sophisticated technique, which only requires basic arithmetic operations. Such a technique can substitute the naive one under the reasonable assumption that the vertices of $P$ be points with algebraic coordinates (as expressed in some global coordinate system, which is not necessarily the local one of any searcher).\footnote{A real number is said to be \emph{algebraic} if it is a root of a polynomial with integer coefficients~\cite{libro1}.}

\paragraph{Statement of results.}
We prove that the Meeting problem in a polygon $P$ can be solved by $k=\sigma+1$ searchers, where $\sigma$ is the order of the rotation group of $P$ (which is also called the \emph{symmetricity} of $P$). We also give a matching lower bound, showing that there are polygons of symmetricty $\sigma$ where $\sigma$ searchers cannot solve the Meeting problem.

Then, since all our lower-bound examples are polygons with a hole around the center, we wonder if the Meeting problem can be solved by fewer searchers if we exclude this small class of polygons (i.e., the polygons with a hole around the center).\footnote{Collectively, these polygons constitute a subset of measure~$0$ of the set of all polygons with holes.} Surprisingly, it turns out that in all the remaining polygons only two searchers are sufficient to solve the Meeting problem. In particular, these include all the polygons with no holes.

Additionally, searchers can geometrically construct their destination points with a compass and a straightedge, provided that the vertices of $P$ are algebraic points. Equivalently, searchers only have to compute combinations of basic arithmetic operations and square root extractions on the coordinates of the visible vertices of $P$. This is done via an encoding technique of independent interest, which we apply to mobile robots for the first time.

As a subroutine of our algorithms, we employ a self-stabilizing map construction algorithm that is of independent interest, as well.

\paragraph{Paper summary.}
In Section~\ref{s:2}, we formally define all the elements of the Meeting problem. In Section~\ref{s:3}, we consider the Meeting problem for searchers equipped with an unlimited amount of persistent internal memory whose initial contents can be arbitrary (hence possibly ``incorrect''). This simplification allows us to present ``cleaner'' versions of our algorithms, which are not burdened by the technicalities of our positional encoding method. In Section~\ref{s:4}, we present our encoding technique and we show how to apply it to the Meeting algorithms of Section~\ref{s:3}, thus extending our results to oblivious searchers. Finally, in Section~\ref{s:5} we discuss directions for further research.

\subsection{Related Work}\label{s:1.3}
\paragraph{Static version.}
If searchers are unable to move, then the Meeting problem becomes equivalent to asking what is the maximum number of points that can be placed in a polygon in such a way that no two of them see each other. This is called the \emph{hidden set problem}, and has been studied in~\cite{hidingpeople}, where tight bounds are given in terms of the number of vertices of the polygon. The situation with mobile searchers is of course radically different.

\paragraph{Gathering in the plane.}
The Gathering problem has been extensively studied in several contexts~\cite{libro2,librosantoro}. The literature can be divided into works considering robots in a geometric space, and works considering agents on a graph~\cite{libro2,approach,pelc1,pelc4}.

Here we focus on Gathering in geometric spaces, since it is related to our setting. In the Look-Compute-Move model, asynchronous oblivious robots with unlimited visibility can solve the Gathering problem without additional assumptions~\cite{item1}.

The case of limited visibility has been studied, as well. Any given pattern can be formed by asynchronous robots with agreement on a coordinate system~\cite{gatheringvisibiliy}. Without such agreement, it is still possible to converge to the same point, perhaps without ever reaching it~\cite{limited1}.

Fault-tolerant Gathering has been investigated in~\cite{rendezvous1,self2} for oblivious robots with unlimited visibility. Rendezvous has been investigated also when the robots have a constant amount of ``visible memory''~\cite{rendezvousviglietta}.

All the aforementioned results hold for robots that inhabit an unbounded plane where no extraneous objects can block visibility or movement. In particular, none of these results pertains robots in a polygon.

\paragraph{Rendezvous in a polygon.}
The work that is most relevant to ours is represented by a series of papers on Rendezvous and approximate Rendezvous by two robots in polygons or more general planar enclosures~\cite{pelc2,pelc3,pelc1,pelc5}. The authors show how to guarantee that the two robots' trajectories will intersect (or get arbitrary close to each other in case of approximate rendezvous) within finite time, in spite of a powerful adversary that controls the speed and the movements of the robots on their trajectories. However, termination happens implicitly: the robots are not necessarily aware of each other's presence, and Rendezvous is considered solved even if they are both moving. Moreover, none of these papers considers oblivious robots, and none of them allows the initial memory contents of the robots to be arbitrary. Both are simplifications of the problem, because they allow robots to implicitly agree on a single landmark and just move there.

In~\cite{pelc1,pelc5}, the authors study the feasibility of approximate Rendezvous by two robots with unique ids (i.e., non-anonymous) in any closed path-connected subset of the plane. Robots have unlimited persistent memory and do not agree on a system of coordinates.

In~\cite{pelc2}, they investigate upper and lower bounds on the movements of two robots solving the Rendezvous problem in a polygon. The authors give a wealth of results under different assumptions, but a common handedness (i.e., a common notion of clockwise direction) and unlimited persistent memory are always assumed.

In~\cite{pelc3}, feasibility conditions for Rendezvous and constructive algorithms are given. Here robots have a constant amount of persistent memory (hence they are not oblivious) and a common handedness.

Further works include~\cite{CDDMW13b,CDDMW15}, addressing the weak rendezvous problem in a polygon, where all robots have to attain mutual visibility at the same time. It is shown that robots can agree on a clique in the visibility graph of the polygon, and the weak rendezvous protocol is to reach a vertex of the elected clique and wait. Unfortunately, such an interesting technique cannot be used in our case. Indeed, we are aiming for a self-stabilizing meeting procedure, which implies that a searcher cannot remain forever in a sub-portion of the polygon. This is clear if we consider an initial memory configuration that forces two searchers to remain in different parts of the polygon (e.g, two disjoint cliques of the visibility graph), each searcher wrongly believing to be in the ``elected portion''. Another drawback of~\cite{CDDMW13b,CDDMW15} is the need for an upper bound on the number of vertices of the polygon.

\paragraph{Miscellanea.}
Another problem for robots in polygons is the search for an intruder, e.g.,~\cite{yamashita2,yamashita1}, where one robot tries to escape while others have to locate it. This setting is clearly quite different, as in the Meeting problem we consider robots that cooperate to achieve a common goal.

The model in which robots can obstruct each other's view has also been studied. Here, the goal is typically to make all robots see each other by making sure that no three of them are collinear~\cite{DFGPSV,SBM15,SVTBR16}. As with the literature on Gathering, none of these works considers robots in a polygon.

Recall that a sub-routine of our Meeting algorithms consists in drawing a map of the polygon. A related problem is that of constructing the visibility graph of a polygon by mobile robots: this has been addressed in~\cite{CDDMW13b,CDDMW15,angles}, where great efforts have been devoted to finding minimal assumptions on the robots' power that allow them to solve the problem. In particular, in \cite{angles} the mapping is done without any a-priori knowledge about the polygon bound, and only using the measurements of angles (hence without measuring distances). However, extending this algorithm to the memory-less case is far from trivial.

\paragraph{Computability issues.}
The issue of defining a model for the ``local'' computations of mobile robots has hardly ever been addressed in the relevant literature. It is nevertheless interesting to establish what destination points are computable by mobile robots, and what it means for a robot to compute a point.

To the best of our knowledge, only two papers deal with this problem. In~\cite{item1}, a definition of computation is not explicitly given, but it is said that a certain point is not computable, in the sense that its coordinates cannot be computed with basic arithmetic operations and extractions of roots of any degree. In~\cite{circle}, computation is explicitly defined in terms of algebraic functions, although all that is actually needed is the ability to construct regular polygons, as well as all points that can be geometrically obtained with compass and straightedge.

Interestingly, solutions to geometric problems have been proposed that involve functions whose computability (in an intuitive sense) is unclear. An example is in~\cite{pelc2}, where the Rendezvous point computed by one of the algorithms is either the central vertex or the midpoint of the central segment of the medial axis of a polygon. Since the central segment of the medial axis could be any parabolic arc, its midpoint is a transcendental function of the polygon's vertices. As such, it is not constructible with compass and straightedge~\cite{kaza}. It turns out that this is not a real problem for that particular algorithm, because the robots can easily agree on an endpoint of the aforementioned parabolic arc or on the parabola's vertex, instead of its midpoint. Still, it is interesting to observe that the notion of computability emerging from~\cite{pelc2} is more ``comprehensive'' than the one of~\cite{circle}.

Another contribution of this paper is a formal definition of the concept of computability for mobile robots (see the beginning of Section~\ref{s:4.2}). Accordingly, all of our geometric constructions can be performed with a compass and a straightedge.

\section{Definitions}\label{s:2}

\paragraph{Polygons.}
A \emph{polygon} in the Euclidean plane $\mathbb R^2$ is a non-empty, bounded, connected, and topologically closed $2$-manifold whose boundary is a finite collection of line segments. The \emph{vertices}, \emph{edges}, and \emph{diagonals} of a polygon are defined in the standard way, as well as the notion of \emph{adjacency} between vertices. One connected component of a polygon's boundary, called the \emph{external boundary}, encloses all others, which are called \emph{holes}.

If a polygon has an axis of symmetry, we say that it is \emph{axially symmetric}. If it has a center of symmetry, we say that it is \emph{centrally symmetrc}. The largest integer $\sigma$ such that rotating a polygon around its barycenter by $2\pi/\sigma$ radians leaves it unchanged is called the \emph{symmetricity} of the polygon.\footnote{Here and throughout the paper, we will refer to the barycenter only because it is a well-defined point in every polygon. However, this choice is not essential: equivalently, we could take the center of symmetry if the polygon has one, or any point otherwise.} In other words, the symmetricity is the order of the rotation group of the polygon. If $\sigma>1$, the polygon is said to be \emph{rotationally symmetric}. (Most of these concepts are introduced, for instance, in~\cite{geometry}.)

We say that a point $p\in P$ \emph{sees} a point $q\in P$ (or, equivalently, that $q$ is \emph{visible} to $p$) if the line segment $pq$ lies in $P$. If, in addition, no vertices of $P$ lie in the relative interior of $pq$, then $p$ \emph{fully sees} $q$ (equivalently, $q$ is \emph{fully visible} to $p$), as Figure~\ref{f:a} illustrates.

\begin{figure}[ht]
\centering
\includegraphics[scale=1.25]{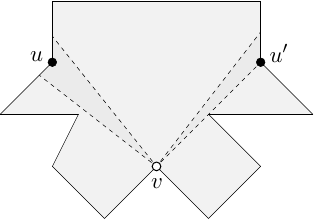}
\caption{$v$ can fully see $u$, and it can see $u'$ but not fully.}\label{f:a}
\end{figure}

\paragraph{Searchers.}
Let $P$ be a polygon. By \emph{$P$-searcher} we mean an anonymous \emph{robot} represented by a point in $P$, which, informally, can \emph{observe} its surroundings and \emph{move} within $P$. If a $P$-searcher is located in $p\in P$, we say that it \emph{sees} all the points of $P$ that are visible to $p$. When $P$ is understood, we will omit it and simply refer to \emph{searchers}.

The \emph{life cycle} of a $P$-searcher consists of three phases, which are repeated forever: \emph{Look}, \emph{Compute}, and \emph{Move}. In a Look phase, the $P$-searcher takes a ``snapshot'' of the subset of $P$ that it currently sees, along with the locations of the $P$-searchers that it currently sees. The snapshot is expressed in the \emph{local reference system} of the observing searcher, which is a Cartesian system of coordinates with the searcher's current location as the origin. In a Compute phase, the searcher executes a deterministic algorithm whose input is the last snapshot taken, and the output is a \emph{destination point}, again expressed in the local reference system of the searcher. In the Move phase, the searcher continuously moves toward the destination point it just computed. Once it gets there, it stops moving and starts a new Look phase, and so on. A searcher's local coordinate system translates as the searcher moves (to keep the searcher's location at the origin), but it retains its orientation, scale, and handedness.\footnote{The fact that a searcher retains its local reference frame's orientation, scale, and handedness is common in most of the related literature. However, in this work we will not strictly need it: in our algorithms, a searcher will always move (close) to a vertex of $P$. Hence, after a move, it will always be able to correctly match its new view with the previous one. This is possible even if its reference frame has reflected after the move: it is sufficient to make the searcher stop close enough to the angle bisector stemming from the destination vertex, but not exactly on it. On the next turn, the searcher will be able to tell its new reference frame's handedness based on whether it is located to the left or the right of the angle bisector (cf.~Section~\ref{s:4.2}).}

When several $P$-searchers are present, we stipulate that they all execute the same algorithm. Furthermore, each searcher's local coordinate system is independent of the others, and has its own orientation, scale, and handedness. Therefore, each $P$-searcher sees a differently scaled, rotated, translated, and possibly reflected version of $P$. Searchers are also \emph{asynchronous}, in the sense that their life cycles are completely independent: each phase of each cycle of each searcher lasts an unpredictably long (but finite) time, which is decided by an adversarial \emph{scheduler}. Also, the speed of a searcher during a Move phase is not necessarily constant. We stress that $P$'s vertices are ``anonymous'', in the sense that they can be distinguished only by their relative positions, and no labels are attached to them.

\paragraph{Meeting.}
We say that two $P$-searchers are \emph{mutually a\-ware} at some point in time if they have seen each other during their most recent Look phases. That is, if searcher $s_1$ sees searcher $s_2$ during a Look phase at time $t_1$, $s_2$ sees $s_1$ during a Look phase at time $t_2\geq t_1$, and neither $s_1$ nor $s_2$ performs another Look phase in the time interval $(t_1,t_2)$, then $s_1$ and $s_2$ are mutually aware at time $t_2$ (and they remain mutually aware until $s_1$ performs a Look phase without seeing $s_2$, or vice versa). A very similar notion of mutual awareness has been defined in~\cite{nearg}.

Given a team of $P$-searchers, the \emph{Meeting problem} prescribes that at least two of them become mutually aware. More precisely, the Meeting problem for $k$ searchers in $P$ is \emph{solvable} if there exists an algorithm $A$ such that, if all $k$ searchers execute $A$ during all their Compute phases, at least two of them eventually become mutually aware, regardless of how the searchers are initially laid out in $P$, and regardless of how the scheduler decides to control their behavior. Occasionally, we will say that two searchers \emph{meet}, as a synonym of becoming mutually aware.

In Section~\ref{s:3}, we are going to assume that each search\-er has an unlimited amount of \emph{persistent internal memory}, which can be read and updated by the searcher during each Compute phase, and is retained for use in later Compute phases. The initial contents of the internal memory of each searcher are arbitrary, and possibly ``incorrect''. In Section~\ref{s:4}, we will drop the persistent memory requirements, and we will extend our algorithms to \emph{oblivious} searchers, whose computations only rely on the single snapshot taken in the most recent Look phase, and whose internal memory is erased during each Move phase.

\section{Algorithms and Correctness}\label{s:3}

In this section, we set out to minimize the number of $P$-searchers that can solve the Meeting problem in a polygon $P$. In Section~\ref{s:3.1}, we provide a tight bound in terms of $P$'s symmetricity, by means of a lower-bound construction (Theorem~\ref{t:1}) and an algorithm which, as a bonus, is independent of $P$ (Theorem~\ref{t:2}). As a tool, we use a self-stabilizing map construction algorithm. In Section~\ref{s:3.2}, we exclude a pathological class of polygons, and we prove that in all remaining polygons the Meeting problem can be solved by just two searchers (which is obviously optimal), again with an algorithm independent of the polygon (Theorem~\ref{t:3}).

We first present our algorithms assuming that searchers have an unlimited amount of persistent internal memory, which initially may contain arbitrary data. Then, in Section~\ref{s:4}, we will extend these algorithms to oblivious searchers.

\subsection{General Algorithm}\label{s:3.1}
\paragraph{Lower bound.}
First we give a lower bound on the minimum number of searchers required to solve the Meeting problem in a polygon. Our bound is in terms of the polygon's symmetricity.

\begin{theorem}\label{t:1}
For every integer $\sigma>0$, there exists a polygon with symmetricity $\sigma$ in which $\sigma$ (or fewer) searchers cannot solve the Meeting problem.
\end{theorem}
\begin{proof}
If $\sigma=1$, the statement is trivial. If $\sigma>1$, we construct a polygon with symmetricity $\sigma$ shaped as a $\sigma$-pointed star with one large hole almost touching the external boundary, as shown in Figure~\ref{f:1}. We then arrange $\sigma' \leq \sigma$ searchers and orient their local coordinate systems in a symmetric fashion, as in Figure~\ref{f:1}. Now, let the initial memory contents of all the searchers be equal, and suppose that the scheduler always activates them synchronously. By the rotational symmetry of our construction, each searcher gets an identical snapshot of the polygon, and therefore all searchers compute symmetric destination points and modify their memory in the same way. This holds true at every cycle, and so, by induction, the searchers will always be found at $\sigma'$ symmetric locations throughout the execution. Note that our polygon has the property that no two of its points whose angular distance (with respect to the barycenter) is a multiple of $2\pi/\sigma$ can see each other. Hence, no matter what algorithm the searchers are executing, no two of them will ever be mutually aware.\qed
\end{proof}

\begin{figure*}[ht]
\centering
\includegraphics[scale=1.25]{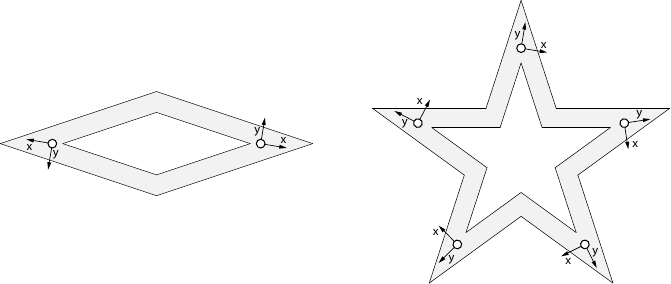}
\caption{Constructions used in Theorem~\ref{t:1} for $\sigma=2$ and $\sigma=5$}\label{f:1}
\end{figure*}

Note that the above theorem holds for searchers with memory, and hence \emph{a fortiori} it holds for oblivious searchers.

Next we will prove that the bound of Theorem~\ref{t:1} is tight, and hence $\sigma+1$ searchers are optimal. Our Algorithm~\ref{f:2} is illustrated below. We have some persistent variables and the procedure Compute, which takes as input the current snapshot, i.e., the part of polygon that is currently visible to the executing searcher plus the searchers it currently sees. This algorithm assumes that searchers have unlimited memory, and hence they can store the entire history of the snapshots they have taken since the beginning of the execution. In Section~\ref{s:4}, we will show how to drop this requirement and apply our algorithms to oblivious searchers.

\begin{algorithm*}
\scriptsize
\begin{framed}
\begin{algorithmic}
\State {\sf Persistent variables}
\State SnapshotList
\State Action
\State Direction
\State Polygon
\State PivotPoint
\\
\State {\sf Procedure Compute }{(Snapshot)}
\If{Snapshot contains no other searcher}
\State Append Snapshot to SnapshotList
\If{SnapshotList is inconsistent \textbf{or} (Action = PATROL \textbf{and} PivotPoint is not consistent with Polygon)}
\State SnapshotList := Snapshot
\State Action := EXPLORE
\EndIf
\If{Action = EXPLORE}
\State Polygon := Extract (partial) polygon from SnapshotList
\State $U$ := Unvisited vertices of Polygon
\If{$U\neq\varnothing$}
\State $v$ := First vertex of $U$
\State Compute a shortest path to $v$ within Polygon, and move to the last visible point along this path
\Else
\State Action := PATROL
\State Direction := CLOCKWISE
\State $S$ := Set of axes of symmetry of Polygon
\If{$S=\varnothing$}
\State $C$ := Select a rotation class of vertices of Polygon in a similarity-invariant way
\State PivotPoint := Select any vertex in $C$
\Else
\State $S'$ := Select a class of equivalent axes in $S$ in a similarity-invariant way
\State $\ell$ := Select any axis in $S'$
\State $C$ := Select a class of equivalent points of $\ell$ on the boundary of Polygon in a similarity-invariant way
\State PivotPoint := Select any point in $C$
\EndIf
\State Augment Polygon using PivotPoint as pivot in a similarity-invariant way to make it simply connected
\EndIf
\EndIf
\If{Action = PATROL}
\If{I am in PivotPoint}
\State Invert Direction
\EndIf
\State Move to the next vertex of Polygon, following its boundary in the direction stored in variable Direction
\EndIf
\EndIf
\end{algorithmic}
\end{framed}
\caption{Meeting algorithm for general polygons\label{f:2}}
\end{algorithm*}

\paragraph{Checking for other searchers.}
By definition, the Meeting problem is solved when two searchers become mutually aware. So, in our algorithm, whenever a searcher $s_1$ sees another searcher $s_2$, it stays idle for a cycle and waits to be noticed by $s_2$. Then, if $s_1$ no longer sees $s_2$, it realizes that they are not mutually aware, and resumes the algorithm (this may happen if $s_2$ is in the middle of a Move phase when it is seen by $s_1$, and it goes through an area that is invisible to $s_1$ before performing its next Look phase). Otherwise, $s_1$ and $s_2$ become mutually aware, and the Meeting problem is solved.

\paragraph{Checking for incongruities.}
Let $P$ be the polygon in which the searchers are located. Since the initial memory contents may be incorrect, if a searcher notices a discrepancy between the current snapshot of $P$ and the history of snapshots stored in memory, it forgets everything and restarts the execution from wherever it is. Note that a searcher can always reconstruct all its previous movements within $P$ by looking at the history of snapshots and ``simulating'' procedure Compute on all of them. Therefore, when the searcher ``believes'' to be re-visiting some region of $P$, it can compare the new snapshot with the old ones taken from the same region, and is able to tell if something looks different. If this is the case, it must be because its initial memory contents were ``corrupt'', and hence it overwrites everything with the current snapshot.

\paragraph{Exploring the polygon.}
We observe that each searcher must keep re-visiting every part of the boundary of $P$. Indeed, if it stops visiting some parts of the boundary, it can never be sure that the shape of $P$ is actually the one it has in memory, and it is easy then to prove that the algorithm cannot solve the Meeting problem (revisiting the boundary of $P$ is what makes the map construction subroutine self-stabilizing).

Our main algorithm is divided into two phases: EXPLORE and PATROL. Roughly speaking, in the EXPLORE phase, a searcher visits all vertices of $P$; in the PATROL phase, it moves back and forth along the boundary of $P$, searching for a companion. The EXPLORE phase is relatively simple: as the searcher explores new vertices, it keeps track of the ones that it has seen but not visited. Then it picks the first of such vertices and moves to it along a shortest path. Since the searcher may not have a complete picture of $P$ yet, by ``shortest path'' we mean a shortest path in the portion of $P$ that has been recorded in memory thus far.

\paragraph{Selecting the pivot point.}
For the PATROL phase, the searcher must first choose a \emph{pivot point} of $P$, which is the point where the searcher changes direction as it patrols $P$'s boundary. It also has to cope with the fact that the boundary of $P$ may not be connected, since $P$ may have holes. The pivot point is chosen in a different way depending if $P$ is axially symmetric or not.

Let $n$ be the number of vertices of $P$, and let $\sigma$ be its symmetricity. Suppose first that $P$ is not axially symmetric. Then, the orbit of each vertex under the rotation group of $P$ has size exactly $\sigma$, and therefore there are $n/\sigma$ different orbits (or \emph{rotation classes}) of vertices. The searcher will pick one rotation class of vertices in a \emph{similarity-invariant} way. This means that the selection algorithm should not depend on the scale, rotation, position, and handedness of $P$, but it should be a deterministic algorithm that only looks at angles between vertices and ratios between segment lengths.\footnote{As an example, we show how to do it when $P$ has no holes. Extending this method to the general case is just slightly more complicated, but the principles are the same. Pick the (unique) circle of smallest radius that contains all the vertices of $P$, and let $r$ be its radius. Name the vertices of $P$ $v_0$, $v_1$, \dots, $v_{n-1}$ in clockwise order. Pick any vertex $v_i$, and construct the right-handed coordinate system having origin in $v_i$, unit $r$, and $x$ axis oriented like $\overrightarrow{v_iv_{i+1}}$ (indices are always taken modulo $n$). Give a representation of $P$ in this coordinate system, i.e., the ordered list of the $x$ and $y$ coordinates of $v_{i+1}$, \dots, $v_n$, $v_1$, \dots, $v_{i-1}$. Then construct another representation in the same coordinate system, but taking the vertices in the reverse order, i.e., $v_{i-1}$, \dots, $v_1$, $v_n$, \dots, $v_{i+1}$. Pick the lexicographically smaller of these two representations (if they are equal, pick any of them), and call it $R_i$. Repeating the same construction with all the $v_i$'s yields the representations $R_0$, $R_1$, \dots, $R_{n-1}$: let $R_m$ be the lexicographically smallest among them. Now, pick all vertices $v_i$ such that $R_i=R_m$: these constitute a rotation class of $P$ chosen in a similarity-invariant way. Indeed, no matter how we rotate, translate, uniformly scale by a non-zero factor, or reflect $P$, we will always pick the same set of vertices.} This is to guarantee that all searchers that have a correct picture of $P$ in memory (expressed in their respective local coordinates systems) will select the same class of vertices. Once this rotation class has been selected, the searcher picks any of its elements as the pivot point.

Suppose now that $P$ is axially symmetric: hence it has $\sigma$ distinct axes of symmetry. If $\sigma$ is odd, all axes of symmetry are rotationally equivalent (i.e., for any two axes of symmetry, there is a rotation of the plane that maps one into the other and leaves the polygon unchanged); if $\sigma$ is even, there are two distinct classes of rotationally equivalent axes of symmetry, each of size $\sigma/2$. (For instance, the pentagonal star in Figure~\ref{f:1} has five equivalent axes of symmetry, while the polygon in Figure~\ref{f:3} has symmetricity 4 and two pairs of equivalent axes of symmetry.) The searcher will select a class of axes of symmetry in a similarity-invariant way, and then pick any axis $\ell$ in this class. The pivot point will then be a point in the intersection between $\ell$ and the boundary of $P$. If $\sigma$ is odd, all such points are distinguishable, so any one of them is chosen by the searcher in a similarity-invariant way. If $\sigma$ is even, these points come in symmetric pairs along $\ell$ (see Figure~\ref{f:3}). In this case, one such pair is selected in a similarity-invariant way, and then any point in the pair is picked by the searcher as the pivot point. Note that the pivot point is either a vertex of $P$ or the midpoint of an edge.

\paragraph{Augmenting the polygon.}
Once a searcher has selected a pivot point, it adds some ``artificial'' edges to $P$ in order to make it simply connected, i.e., remove all its holes. This may be impossible to do in a similarity-invariant way (for instance, in the polygons of Figure~\ref{f:1} it is impossible), so the pivot point will be used to determine how symmetries are broken. Also, we will make sure that no such artificial edges are incident to the pivot point.

If $P$ is not axially symmetric, then an orientation (i.e., clockwise or counterclockwise) can be chosen in a similarity-invariant way. Given the pivot point and an orientation, then breaking symmetries is trivial. In order to remove a hole, we draw a diagonal of $P$ that connects two different connected components of the boundary (i.e., two different holes or a hole and the external boundary). The diagonal is selected in a deterministic way, and should not be incident to the pivot point. Cutting $P$ along such a diagonal merges two connected components of its boundary, reducing the number of holes by one. This procedure is repeated until the boundary is connected.

Suppose now that $P$ is axially symmetric, and let $\ell$ be the axis of symmetry containing the pivot point. We will augment $P$ while keeping it symmetric with respect to $\ell$. If a hole of $P$ intersects $\ell$, we connect it to a neighboring hole or to the external boundary of $P$ in a deterministic way, by drawing a sub-segment of $\ell$ not incident to the pivot point. If a hole $H$ of $P$ does not intersect $\ell$, it must have a symmetric hole $H'$ on the other side of $\ell$. Then we draw a diagonal (again, in a deterministic way) not intersecting $\ell$ to connect $H$ to another hole or to the external boundary. We also draw the symmetric diagonal to connect $H'$. Since the two diagonals do not intersect each other (or they would intersect $\ell$), cutting $P$ along them does not disconnect it. Figure~\ref{f:3} shows an example of how such diagonals can be chosen in a symmetric polygon (in this example, $\ell$ is the vertical axis).

\begin{figure*}[ht]
\centering
\includegraphics[scale=1]{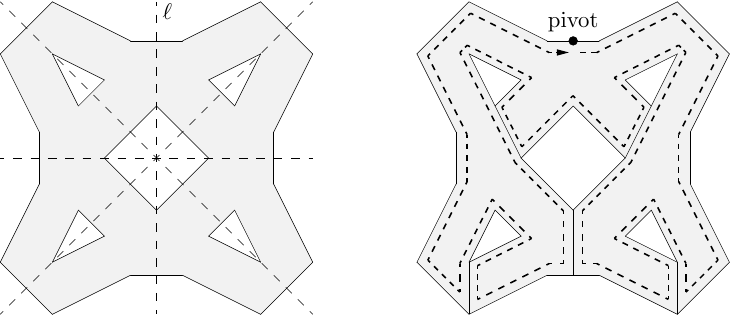}
\caption{Augmenting an axially symmetric polygon and defining a tour of its boundary}\label{f:3}
\end{figure*}

\paragraph{Patrolling the polygon.}
In the previous paragraphs, we described how to select a finite set of segments in $P$: let $D$ be the union of these segments. As a result of cutting $P$ along such segments, we obtain a \emph{degenerate} simply connected polygon $\widetilde{P}=P\setminus D$. By ``degenerate'' we mean that its boundary is no longer the boundary of a topologically closed 2-manifold. However, it is possible to perform a \emph{tour} of the boundary of $\widetilde{P}$, by walking along the external boundary of $P$, and then taking a detour along a segment of $D$ and around a hole of $P$, as soon as one is found. The resulting tour can be clockwise or counterclockwise, and traverses each edge of $P$ once and each segment of $D$ twice (once in each direction). One such tour is illustrated in Figure~\ref{f:3}. Intuitively, this would correspond to slightly ``thickening'' each segment of $D$, subtracting $D$ from $P$, and walking around the boundary of the resulting (non-degenerate) polygon.

The PATROL phase of our algorithm consists in taking a tour of $\widetilde P$ and switching direction (from clockwise to counterclockwise and vice versa) every time the pivot point is reached. So, all vertices of $P$ are perpetually visited in some fixed order, then in the opposite order, and so on. At any time, the searcher can always determine its next destination point based on the history of snapshots stored in memory.

\paragraph{Correctness of Algorithm~\ref{f:2}.}
We will now prove the correctness of this algorithm.

\begin{theorem}\label{t:2}
There is an algorithm that, for every integer $\sigma>0$, solves the Meeting problem with $\sigma+1$ searchers (regardless of their initial memory contents) in every polygon with symmetricity $\sigma$.
\end{theorem}
\begin{proof}
We will show that Algorithm~\ref{f:2} correctly solves the Meeting problem for $\sigma+1$ searchers in any polygon $P$ with symmetricity $\sigma$. We have to show that, as the searchers execute the algorithm (asynchronously), at least two of them will eventually become mutually aware, regardless of the initial memory contents of the searchers and their initial locations.

Since the initial memory contents of a searcher may be incorrect, when a searcher notices a discrepancy between the current observation and a previous observation, it erases its own memory and restarts the execution. The same happens if it realizes that the pivot point it has chosen does not match the polygon. From that point onward, the searcher's memory will only contain correct information, and the execution will never be restarted again. Hence, in the following, we will assume that no such discrepancy is ever discovered, and therefore the execution is never restarted.

The EXPLORE phase relies on the connectedness of the \emph{visibility graph} of $P$. Recall that the visibility graph of $P$ is the graph on the set of vertices of $P$ whose edges are the edges and diagonals of $P$. This graph is connected because from any vertex to any other vertex there is a shortest path that is a polygonal chain turning only at (reflex) vertices of $P$. So, as the searcher walks through the visibility graph, it maintains a list of vertices that have been discovered but not visited. It then walks to the first of these vertices along a shortest path while updating the list, and so on. Note that the ``shortest path'' may change as new vertices are discovered. However, this can only happen finitely many times, and eventually the target vertex is indeed reached. So, the list of discovered but unvisited vertices will eventually be depleted. By the connectedness of the visibility graph, this happens if and only if all vertices have been visited. This means that eventually the searcher will have a complete representation $P'$ of the polygon $P$. Recall that $P$ and $P'$ may not be the same polygon, because the searcher may have an arbitrary list of snapshots initially in memory, which may be coherent with the current snapshot.

Now that the searcher has a representation $P'$ of $P$, it makes its boundary connected by choosing a pivot point and adding some extra segments, and then starts the PATROL phase. Observe that the pivot point and the extra segments remain fixed as the searcher moves, since they have been stored in the persistent memory. In the PATROL phase, the searcher will repeatedly attempt to visit every vertex of $P'$. So, if $P\neq P'$, the searcher will eventually find out: if some vertices of $P'$ are not vertices of $P$ or if $P$ has some extra vertices, the searcher is bound to see the discrepancy, again due to the connectedness of the visibility graph. But this contradicts our assumptions, hence we may as well assume that $P=P'$.

We can therefore assume without loss of generality that, at some point, all searchers are in the PATROL phase, they all have a correct representation of $P$ in memory, and they have correctly computed a pivot point and correctly augmented $P$ to make it simply connected. Suppose that $P$ is not axially symmetric. Since the rotation class of vertices to which the pivot point belongs is chosen in a similarity-invariant way by all searchers, they all have picked the same class. Hence there are only $\sigma$ possible choices for the pivot point, and two searchers must have picked the same, by the pigeonhole principle. Suppose now that $P$ is axially symmetric, and hence it has $\sigma$ axes of symmetry. If $\sigma$ is odd, two searchers must have picked the same axis of symmetry, say $\ell$. These two searchers have then identified a pivot point on $\ell$ in a similarity-invariant way, and therefore they have picked the same point. If $\sigma$ is even, there are two classes of equivalent axes, each of size $\sigma/2$. All searchers have picked an axis from the same class, and hence three searchers must have picked the same axis, say $\ell$, by the pigeonhole principle. Then, each of these three searchers has chosen one of two equivalent points of $\ell$, and therefore two searchers have chosen the same point.

In any case, there are two searchers $s_1$ and $s_2$ that have the same pivot point. These two searchers will also compute the same augmented polygon $\widetilde{P}$, because this is done in a similarity-invariant way (even if $P$ is axially symmetric and $s_1$ and $s_2$ do not have the same notion of clockwise direction). So, both searchers will perform a clockwise tour of the boundary of $\widetilde{P}$, touching all of its vertices is some fixed order, followed by a counterclockwise tour, touching all vertices in the opposite order, and so on. Since they both turn around at the same pivot point, they do the same tour. As a consequence, by the time one of them has completed a full tour, they will have to traverse the same edge $e$ of $\widetilde P$ in opposite directions at the same time. So, they will become mutually aware when reaching the endpoints of $e$, solving the Meeting problem. (As a special case, they may reach the same vertex of $\widetilde P$ at the same time, and then they immediately become mutually aware.)

There is one last detail to consider. Recall that a searcher $s_1$ remains idle for a cycle whenever it sees another searcher $s_2$, even if $s_2$ is not going to notice $s_1$. This may happen, for instance, if $s_2$ is traveling between two points that cannot see $s_1$'s location. If this situation keeps repeating every time $s_1$ takes a snapshot, then $s_1$ is stuck forever, unable to explore or patrol the polygon, and perhaps unable to ever become mutually visible with any other searcher. However, not all searchers can remain stuck in the aforementioned way without at least two of them being mutually aware. Hence, even if $s_1$ is stuck forever and the Meeting problem is not solved yet, at least one searcher necessarily makes steady progress in the algorithm, becoming mutually aware with $s_1$ by the time it completes a full tour of the polygon.\qed
\end{proof}

We emphasize that, if a searcher were tasked to construct a map of $P$, it could do so by simply executing the above algorithm indefinitely (i.e., ignoring the presence of other searchers). Since the algorithm eventually discovers and corrects any possible inconsistency in the initial memory state of the searcher, it is self-stabilizing.

\paragraph{The importance of exploring holes.}
The reader may wonder why we chose to include the holes as part of the tour of the boundary of $P$ that the searchers perform in the PATROL phase. Indeed, the searchers could easily identify the external boundary of $P$ (by computing the sign of its total curvature), so it would be tempting to let them patrol only that part of the boundary, ignoring the holes. This, however, may not work if the initial memory contents of the searchers are incorrect. Say $P$ is not rotationally symmetric, but suppose that it looks rotationally symmetric from the external boundary. This may be because it has a small irregular central hole that is hidden from the external boundary by other holes, while everything else is rotationally symmetric, as in Figure~\ref{f:b}. Since $P$ is not rotationally symmetric, two $P$-searchers should be able to select the same pivot point, and hence meet as they patrol the external boundary. However, their internal representation of $P$ may be incorrect, and show a polygon $P'$ that is rotationally symmetric and coincides with $P$ as seen from the external boundary. So, the searchers may actually choose different pivot points and never notice any discrepancy between $P$ and $P'$ as they patrol the external boundary. But then, they may fail to meet if they occupy symmetric locations and the scheduler keeps activating them synchronously, as explained in Theorem~\ref{t:1}.

\begin{figure}[ht]
\centering
\includegraphics[width=\linewidth]{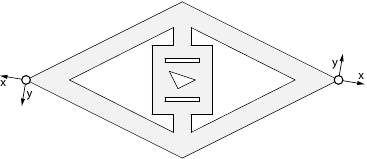}
\caption{The polygon has symmetricity $1$, but its symmetricity looks $2$ if it is observed from the external boundary. The searchers cannot meet if they do not explore the holes.}\label{f:b}
\end{figure}

\subsection{Improved Algorithm for Polygons with Barycenter not in a Hole}\label{s:3.2}

Recall that the worst-case examples given in Theorem~\ref{t:1} are polygons with a hole around the barycenter. It is natural to wonder if the Meeting problem can be solved with fewer searchers if we exclude this special type of polygons. It turns out that in all other cases Algorithm~\ref{f:2} can be drastically improved: only two searchers are needed whenever the polygon's barycenter is not in a hole. Notably, this includes all polygons with no holes.

\paragraph{Counterexample.}
Observe that simply making the searchers patrol the boundary of the polygon as in the previous algorithm may not solve the Meeting problem, even if the polygon has no holes. For instance, assume that the polygon has symmetricity $4$ and has a central region with four equal branches, shaped in such a way that a searcher that is far enough inside a branch cannot see any of the central region, as in Figure~\ref{f:c}. Suppose that two searchers are patrolling this polygon, and they have different pivot points. Then, the scheduler can always keep them in different branches of the polygon and make them move symmetrically within their respective branches (recall that they are executing the same deterministic algorithm). When they have to move to the next branch, the scheduler will make one searcher quickly move to the central region and into the new branch while the other searcher remains hidden inside its own branch. Then the scheduler will make the second searcher move through the central region while the first one is hidden. This way, the searchers will never meet.

\begin{figure*}[ht]
\centering
\includegraphics[scale=1]{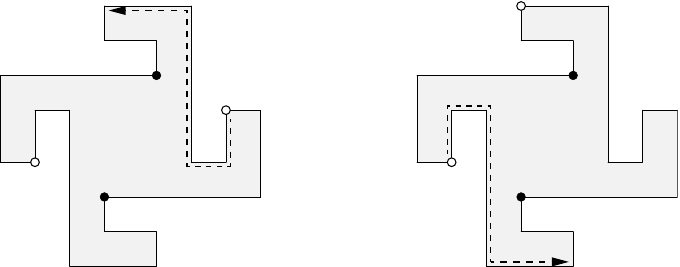}
\caption{If two searchers patrol the boundary with different pivot points, they may never meet.}\label{f:c}
\end{figure*}

Our improved Meeting algorithm is given below as Algorithm~\ref{f:4}. It begins by testing for the presence of another searcher, followed by some consistency tests, and an EXPLORE phase, which are essentially the same as in the previous algorithm. It then proceeds with a PATROL phase, which is more complex than the old one. Note that Algorithm~\ref{f:2} already solves the Meeting problem with two searchers if the polygon is not rotationally symmetric (i.e., for $\sigma=1$). So, in this special case, our improved algorithm works exactly as the previous one. In the following, we will therefore assume that the polygon is rotationally symmetric, and we will discuss only the new PATROL phase.

\begin{algorithm*}
\scriptsize
\begin{framed}
\begin{algorithmic}
\State {\sf Persistent variables}
\State SnapshotList
\State Action
\State Stage
\State Polygon
\State PivotVertex
\State PolygonTriangles
\State PolygonLevels
\\
\State {\sf Procedure Compute }{(Snapshot)}
\If{Snapshot contains no other searcher}
\State Append Snapshot to SnapshotList
\If{Persistent variables are inconsistent}
\State SnapshotList := Snapshot
\State Action := EXPLORE
\EndIf
\If{Action = EXPLORE}
\State Polygon := Extract (partial) polygon from SnapshotList
\State $U$ := Unvisited vertices of Polygon
\If{$U\neq\varnothing$}
\State $v$ := First vertex of $U$
\State Compute a shortest path to $v$ within Polygon, and move to the last visible point along this path
\Else
\State Action := PATROL
\State Stage := $-1$
\If{Polygon is rotationally symmetric}
\State $C$ := Select a similarity class of vertices of Polygon closest to the center in a similarity-invariant way
\State PivotVertex := Select any vertex in $C$
\State Augment Polygon in a similarity-invariant way to make it simply connected
\State Triangulate each branch of augmented Polygon in a similarity-invariant way
\State PolygonTriangles := Total number of triangles in the triangulation of augmented Polygon
\State PolygonLevels := Height of the dual tree of the triangulation of each branch of augmented Polygon
\Else
\State PivotVertex := Select a vertex of Polygon in a similarity-invariant way
\EndIf
\EndIf
\EndIf
\If{Action = PATROL}
\If{Polygon is rotationally symmetric}
\If{I am in PivotVertex}
\State Stage := $\rm{Stage}+1$
\If{Stage $\geq 2\cdot\rm{PolygonLevels} + 2\cdot\rm{PolygonTriangles}^2$}
\State Stage := $0$
\EndIf
\EndIf
\If{Stage = $-1$}
\State Move to the next vertex in a shortest path to PivotVertex
\ElsIf{Stage $<$ PolygonLevels}
\State $j$ := Stage
\State Move to the next vertex of a clockwise $j$-tour of Polygon
\Else
\State $j$ := $2\cdot\rm{PolygonLevels} + 2\cdot\rm{PolygonTriangles}^2 - \rm{Stage}$
\If{$j>$ PolygonLevels}
\State $j$ := PolygonLevels
\EndIf
\State Move to the next vertex of a counterclockwise $j$-tour of Polygon
\EndIf
\Else
\If{I am in PivotVertex}
\State Stage := $\rm{Stage}+1$
\EndIf
\If{Stage is odd}
\State Move to the next vertex of Polygon, following its boundary in the clockwise direction
\Else
\State Move to the next vertex of Polygon, following its boundary in the counterclockwise direction
\EndIf
\EndIf
\EndIf
\EndIf
\end{algorithmic}
\end{framed}
\caption{Improved Meeting algorithm for polygons with barycenter not in a hole\label{f:4}}
\end{algorithm*}

\paragraph{Selecting the pivot vertex.}
Let $P$ be the polygon in which the two searchers operate. Upon ending the EXPLORE phase, a searcher does some pre-processing on the polygon. First it picks a pivot vertex of $P$. To do so, it selects a \emph{similarity class} of vertices $C$ that are closest to the center of the polygon in a similarity-invariant way. A similarity class is a set of vertices that are equivalent up to similarity. This means that both searchers will select the same class of vertices $C$ (assuming they have a correct picture of $P$ in memory). If the symmetricity of $P$ is $\sigma$, then $C$ has size either $\sigma$ or $2\sigma$: indeed, the points of $C$ must be either the vertices of a regular $\sigma$-gon or of two rotated copies of a regular $\sigma$-gon. Each searcher then arbitrarily picks a pivot vertex in this class and stores it in its persistent memory.

\paragraph{Augmenting the polygon.}
The next step is to augment $P$ with some extra edges. Note that the vertices of $C$ (as defined in the previous paragraph) form an equiangular polygon $Q$ around the center of $P$ (a polygon is \emph{equiangular} if all its internal angles are equal). In particular, $Q$ is convex. Since the center of $P$ is not in a hole, $Q$ is completely contained in $P$, i.e., it intersects the boundary of $P$ only at the vertices. We call each of the connected components of $P\setminus Q$ a \emph{branch} of $P$. For each axis of symmetry $\ell$ of a branch that is also an axis of symmetry of $P$, we cut the branch along $\ell$. This operation may merge different connected components of the boundary of $P$, reducing the number of its holes. However, it is easy to see that it cannot disconnect $P$, because we cut only along axes of symmetry, and we leave the central area $Q$ uncut.

If some holes are remaining in the branches, we resolve them by further cutting $P\setminus Q$ along some segments, chosen in a similarity-invariant way, whose endpoints are collinear with the center of $P$. We do so without disconnecting any branch. Note that, since these segments are ``radial'', they cannot intersect each other or the axes of $P$.

The resulting degenerate polygon $\widetilde P$ has simply connected interior and has the same axes of symmetry and the same symmetricity as $P$. Moreover, any searcher performing the above operations on $P$ obtains the same $\widetilde P$, because everything is computed in a similarity-in\-var\-iant way.

\paragraph{Triangulating the branches.}
Each connected component of $\widetilde P\setminus Q$ is called a \emph{sub-branch} of $P$. So, each branch either coincides with a sub-branch or is divided by an axis of symmetry of $P$ into two \emph{twin} sub-branches. As a final pre-processing step, each sub-branch of $P$ is triangulated in a similarity-invariant way. This means, in particular, that twin sub-branches are triangulated in symmetric ways. The central polygon $Q$ is not triangulated.

The dual graph of the triangulation of each sub-branch is a tree. If we add a root node corresponding to $Q$ and we attach all these trees to it, we obtain a rooted tree that is the dual of the entire partition of $\widetilde P$. We denote the height of this rooted tree by $m$.

Figure~\ref{f:5} shows the result of the above operations on an axially symmetric and centrally symmetric. polygon with holes. In this example, the symmetricity is $4$, the branches are four, the sub-branches are eight, and $m=8$.

\begin{figure*}[ht]
\centering
\includegraphics[scale=1]{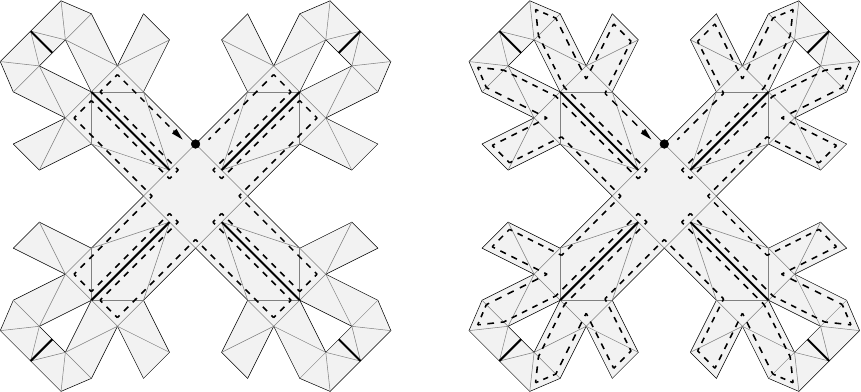}
\caption{Augmented and triangulated axially symmetric polygon with a $3$-tour and a $6$-tour. Solid thick segments represent the cuts that are made to augment the polygon.}\label{f:5}
\end{figure*}

\paragraph{Patrolling the polygon.}
Once $P$ has been augmented and its sub-branches have been triangulated, the PATROL phase starts. This phase has a ``primitive'' operation called \emph{$j$-tour}, where $j$ is an integer between $0$ and $m$. Let $P_j$ be the union of $Q$ and the triangles of the triangulation whose corresponding nodes of the dual graph have depth at most $j$ (with respect to the root corresponding to $Q$). So, for instance, $P_0=Q$ and $P_m=P$. A $j$-tour is a tour of the boundary of $P_j$, starting and ending at the pivot vertex, following the edges of $\widetilde P$. For example, a $0$-tour is simply a tour of the boundary of $Q$, an $m$-tour is a tour of the boundary of $\widetilde P$ (much like the tours of Algorithm~\ref{f:2}), and Figure~\ref{f:5} shows a $3$-tour and a $6$-tour. Obviously, a searcher can perform a $j$-tour in two different directions: clockwise or counterclockwise. In the following, when we say ``clockwise'' and ``counterclockwise'', we mean it in the local reference system of the searcher executing the algorithm.

The PATROL phase consists of several \emph{stages}, and in each stage the searcher performs a $j$-tour, for some $j$. The $j$-tours are performed according to the following list, which is repeated until the Meeting problem is solved:
\begin{itemize}
\item a clockwise $0$-tour,
\item a clockwise $1$-tour,
\item a clockwise $2$-tour,
\item \dots
\item a clockwise $(m-1)$-tour,
\item a sufficiently large number of counterclockwise $m$-tours (twice the square of the total number of triangles in the triangulations of all the sub-branches of $P$ is abundantly enough),
\item a counterclockwise $(m-1)$-tour,
\item a counterclockwise $(m-2)$-tour,
\item \dots
\item a counterclockwise $1$-tour.
\end{itemize}
The first $m$ stages, where the searcher performs clockwise $j$-tours, are called \emph{ascending stages}. All the other stages are called \emph{descending stages}. Moreover, the first stage is called the \emph{central stage}, and the stages in which an $m$-tour is performed are called \emph{perimeter stages}. So, the central stage is an ascending stage, and the perimeter stages are descending stages.

Recall that two different searchers executing the algorithm may not have the same notion of clockwise direction, and therefore in their respective ascending stages they may actually perform tours in opposite directions. If two searchers have the same notion of clockwise direction, they are said to be \emph{concordant}; otherwise, they are \emph{discordant}.

\paragraph{Correctness of Algorithm~\ref{f:4}.}
We can now proceed with the proof of correctness of this algorithm.

\begin{remark}
Similar to the algorithm of Section~\ref{s:3.1}, this one also makes a searcher stop when it sees the other searcher. However, this cannot cause one of them to remain stopped indefinitely without being seen by the other searcher, as already explained in the last paragraph of the proof of Theorem~\ref{t:2}. Therefore, for brevity, in the following proofs we will omit to mention this aspect.
\end{remark}

\begin{lemma}\label{l:1}
Let two $P$-searchers be executing Algorithm~\ref{f:4}, let both be in the PATROL phase, and let both have a correct representation of the polygon $P$ in memory, which is rotationally symmetric. Then, the searchers will either become mutually aware or be in a perimeter stage at the same time.
\end{lemma}
\begin{proof}
Assume that the searchers never become mutually aware. Then, at some point in time, a searcher $s_1$ must start a full series of perimeter stages. If, at this point, the other searcher $s_2$ is also in a perimeter stage, there is nothing to prove. So, let us assume that $s_2$ is not in a perimeter stage. Now, $s_1$ will perform a full series of perimeter stages, repeatedly following the boundary of $\widetilde P$, and touching the central polygon $Q$ and every triangle of the triangulation at each stage. This means that, in the time $s_1$ performs one complete perimeter stage, $s_2$ cannot remain in the same triangle of the partition (or on its boundary), because otherwise it certainly meets $s_1$.

Let $t$ be the total number of triangles in the triangulation of the sub-branches of $P$, and let $T_j$ be the set of such triangles that are in $P_j$. Recall that a $j$-tour, for $j>0$, is a tour of the perimeter of $P_j$. Note that each triangle of $T_j$ has either one edge or two consecutive edges on the boundary of $P_j$. It follows that, as $s_2$ performs a $j$-tour, it moves from one triangle of $T_j$ to another at most $|T_j|\leq t$ times. For $j=0$, the same is trivially true: $s_2$ touches at most $t$ triangles in a $0$-tour.

So, every time $s_1$ performs $t$ perimeter stages, $s_2$ must complete at least one stage. The number of non-perimeter stages is $2m-1\leq 2t-1$, which means that after at most $2t^2-t$ perimeter stages of $s_1$, also $s_2$ must start a perimeter stage. When this happens, $s_1$ still has some perimeter stages to perform, because they are $2t^2$ in total. Hence, both $s_1$ and $s_2$ will be found in a perimeter stage at the same time.\qed
\end{proof}

\begin{corollary}\label{c:1b}
Let two discordant $P$-searchers be executing Algorithm~\ref{f:4}, let both be in the PATROL phase, and let both have a correct representation of the polygon $P$ in memory, which is rotationally symmetric. The searchers will eventually become mutually aware.
\end{corollary}
\begin{proof}
Suppose for a contradiction that the two searchers $s_1$ and $s_2$ never become mutually aware. Following the proof of Lemma~\ref{l:1}, we argue that $s_1$ still has some perimeter stages to perform when $s_2$ is finally forced to start its first one. So, they will both do at least one complete $m$-tour in opposite directions (because they are discordant), thus necessarily crossing each other and becoming mutually aware. Note that our proof goes through even if $s_2$ is performing a perimeter stage when $s_1$ begins the first one. Indeed, $s_2$ must move on to a non-perimeter stage before $s_1$ completes its first perimeter stage, or else they would meet. From now on the proof is the same as in Lemma~\ref{l:1}, with the only difference that $s_1$ has at most one less perimeter stage to perform, which is irrelevant (we chose the number of perimeter stages to be much higher than needed).\qed
\end{proof}

\begin{lemma}\label{l:2}
Let two concordant $P$-searchers be executing Algorithm~\ref{f:4}, let both be in the PATROL phase, and let both have a correct representation of the polygon $P$ in memory, which is rotationally symmetric. If one searcher begins a $j$-tour in an ascending stage while the other searcher is performing a $(j+1)$-tour in a descending stage, with $0\leq j<m$, they eventually become mutually aware. Similarly, if one searcher begins a $j$-tour in a descending stage while the other searcher is performing a $(j-1)$-tour in an ascending stage, with $0<j<m$, they eventually become mutually aware.
\end{lemma}
\begin{proof}
We will only discuss the case in which searcher $s_1$ is starting a $j$-tour in an ascending stage while searcher $s_2$ is performing a $(j+1)$-tour in a descending stage. The other case is symmetric and the proof is essentially the same (it is actually simpler, because it does not involve a $j=0$ or a $j=m$ case). We are going to show that, by the time $s_1$ has finished the current stage, it becomes mutually aware with $s_2$. Note that, since the searchers are concordant, their notion of clockwise direction is the same, and we may assume that this notion also agrees with the ``global'' one. So, the searchers are traveling in opposite directions: $s_1$ is ascending (hence going clockwise) and $s_2$ is descending (hence going counterclockwise).

Because $s_1$ has just started a $j$-tour, it will perform a complete clockwise tour of the boundary of $P_j$, while $s_2$ is somewhere in the middle of a counterclockwise tour of the boundary of $P_{j+1}$ and will then proceed with a tour of $P_j$, as well (because $s_2$ is descending). The set difference between $P_{j+1}$ and $P_j$ is a collection of triangles $T$ of the triangulation of the sub-branches of $P$. Each triangle in $T$ has an edge in common with $P_j$. So, as $s_1$ travels around $P_j$, it also gets to see all of $T$, including all the locations in which $s_2$ could be as it performs the $(j+1)$-tour. Therefore, if $s_1$ finishes the $j$-tour before $s_2$ has completed the $(j+1)$-tour (or at the same time), they must become mutually aware. This happens in particular if $s_1$ reaches the pivot vertex of $s_2$ before $s_2$ does (or at the same time).

Suppose now that $s_1$ reaches the pivot vertex of $s_2$ strictly after $s_2$. So, when $s_2$ reaches its pivot vertex, it finishes its $(j+1)$-tour and starts a $j$-tour, while $s_1$ is still performing its $j$-tour. If $j=0$, both searchers are on the boundary of the central polygon $Q$, and so they become mutually aware. If $j>0$, this is again a descending stage for $s_2$, and so the $j$-tour it performs is counterclockwise. Observe that $s_1$ cannot terminate the current stage before reaching the pivot vertex of $s_2$. But since now both searchers are walking on the boundary of $P_j$ in opposite directions, they are bound to bump into each other and become mutually aware.\qed
\end{proof}

\begin{remark}\label{r:1}
Lemma~\ref{l:2} also holds when both searchers start a $j$-tour in opposite directions at the same time, because this can be considered the very end of the second searcher's previous $(j+1)$-tour (or $(j-1)$-tour).
\end{remark}

\begin{theorem}\label{t:3}
There is an algorithm that solves the Meeting problem with two searchers (regardless of their initial memory contents) in every polygon whose barycenter does not lie in a hole.
\end{theorem}
\begin{proof}
We will show that Algorithm~\ref{f:4} correctly solves the Meeting problem for two searchers in any polygon $P$ whose barycenter does not lie in a hole. The proof of correctness is the same as that of Theorem~\ref{t:2}, except for the PATROL phase. Also, as a searcher still visits every vertex of the polygon during the PATROL phase, it still eventually finds out if its memory is inconsistent with $P$, and in that case it restarts the execution. This can happen only once, because afterwards its memory contents are going to be always correct. So, in the following, we will assume that both searchers already have a correct picture of $P$ in memory, and are both in the PATROL phase. Moreover, since the new algorithm works in the same way as the old one if $P$ is not rotationally symmetric (and the proof of correctness is the same as in Theorem~\ref{t:2}), we will assume that $P$ is rotationally symmetric.

If the two searchers are discordant, they must become mutually aware, due to Corollary~\ref{c:1b}. Let us then assume that they are concordant, and that they never become mutually aware. Therefore, by Lemma~\ref{l:1}, they are eventually found in a perimeter stage at the same time. Then, they will perform all the remaining descending stages, followed by the ascending stages, starting with the central stage. If they start the central stage at the same time, they necessarily become mutually aware, because they are on the boundary of the central polygon $Q$, which is convex and empty. So, one searcher must begin the central stage while the other is still in a descending stage. Then, as one searcher ascends and the other descends, the hypotheses of Lemma~\ref{l:2} are going to be satisfied (also due to Remark~\ref{r:1}), which means that the searchers eventually become mutually aware.\qed
\end{proof}

\paragraph{Polygons with even symmetricity.}
Observe that, if a polygon $P$ is centrally symmetric and its center lies in a hole, then two $P$-searchers placed in symmetric locations and activated synchronously will never see each other (regardless of the shape of the hole). Therefore, Theorem~\ref{t:3} yields a characterization of the polygons of even symmetricity in which the Meeting problem can be solved with two searchers.

\begin{corollary}\label{c:4}
If $P$ has even symmetricity, then the Meeting problem for two $P$-searchers is solvable if and only if the barycenter of $P$ does not lie in a hole.\qed
\end{corollary}

\section{Memoryless Implementations}\label{s:4}

The Meeting algorithms given in the previous section assumed that the searchers were able to memorize the entire history of the snapshots they had taken since the beginning of the execution. With a little extra effort, we could have made a more efficient use of memory, and we could have designed equivalent algorithms that used only a number of variables that is linear in the number of vertices of the polygon.

In this section, we are going to do much better: we will show that we can re-implement our algorithms without using any persistent memory at all. So, our searchers will be \emph{oblivious}, in the sense that the destination point computed in each Compute phase will depend only on the snapshot taken in the most recent Look phase, while all previous snapshots and computations are forgotten.

We achieve this in two steps: in Section~\ref{s:4.1}, we will discuss two ways of encoding all the permanent variables as a single real number; in Section~\ref{s:4.2}, we will show how to apply these encoding techniques to our algorithms.

\subsection{Encoding Persistent Variables}\label{s:4.1}

As a first step, we want to be able to encode all the persistent variables used in our algorithms as a single real number. We will briefly discuss a naive approach, which works for every polygon but yields a code that is not computable on a real random-access machine. Then, we will present an improved code that can be computed with basic arithmetic operations but requires the vertices of the polygon to be algebraic points.

\paragraph{Representing snapshots.}
We have used several types of persistent variables in Algorithms~\ref{f:2} and~\ref{f:4}, such as integers, reals, and snapshots. However, since the algorithms are deterministic, only the snapshots are really necessary. If a searcher remembers the history of the snapshots taken during the execution, it can reconstruct at any time all its past computations, including the history of all the modifications to the other persistent variables (recall that the values of these variables are fixed after each memory ``reset'', when a searcher erases its own memory and restarts the execution). So, since the non-snapshot variables are redundant, we will focus on representing snapshots. Up to this point, we have treated snapshots as primitive data types that could somehow be processed by searchers, but now we have to define them exactly in terms of more elementary variables.

Recall that a snapshot is a representation of the visible portion of the polygon plus a list of visible searchers. The visible searchers are not very important in our Meeting algorithms, and do not even have to be stored in the persistent memory of the observing searcher. Their exact locations are irrelevant, as well. In fact, we may assume that each snapshot that a searcher gets as input simply contains a flag indicating the presence or absence of other searchers in the visible area.

The part of snapshot representing the visible portion of the polygon demands more attention: it encodes a sub-polygon of $P$ expressed in the coordinate system of the observing $P$-searcher. This region is fully described by the portion of $P$'s boundary that is seen by the searcher, which in turn is a union of line segments, each of which is a sub-segment of an edge of $P$. So, we can stipulate that a snapshot takes the form of a finite array of real numbers, say $$(x_1, y_1, x'_1, y'_1, x_2, y_2, x'_2, y'_2, \dots),$$ where $(x_i, y_i)$ and $(x'_i, y'_i)$ are the endpoints of the $i$th segment of the portion of $P$'s boundary that is visible to the searcher (note that none of these points is necessarily a vertex of $P$). Snapshots are received as input by the searcher in this form (plus the visible searchers flag defined in the previous paragraph), and they are also represented by the searcher in this form when they are stored in memory (without the visible searchers flag).

It is easy then for the searcher to manipulate this data type in its computations. For instance, it can readily merge different snapshots and eventually construct a full representation of $P$ as a list of its edges.

\paragraph{General idea.}
As an oblivious $P$-searcher has no persistent memory and can only see its current surroundings, the only way it can implicitly memorize information is by carefully positioning itself within $P$. Specifically, suppose that, among the vertices of $P$ that are visible to the searcher, there is a unique vertex $v$ that is closest to it, and let $d$ be their distance in the searcher's coordinate system (recall that different searchers may have difference units of distance). Then, we say that the searcher \emph{encodes} the number $d$, and its \emph{virtual vertex} is $v$. Note that, since $v$ is the closest visible vertex, it is also fully visible to the searcher (cf.~Figure~\ref{f:a}), which is therefore always able to identify it as a vertex of $P$ by examining a snapshot taken from its current location, even if the snapshot is represented as we explained above (hence not explicitly marking the vertices of $P$). Once the searcher has identified $v$, then it can easily retrieve $d$.

So, a $P$-searcher can encode a range of non-negative real numbers that depends on its unit of distance and the shape of $P$. Also, not all virtual vertices allow to encode the same set of values. However, if $d$ can be encoded under some virtual vertex $v$, then any value in the range $[0,d]$ can be encoded, by letting the searcher approach $v$ by a suitable amount.

Since this method only allows a searcher to encode one number at a time, our goal is to ``pack'' a whole list of snapshots into a single non-negative real number. We would also like to define our packing in such a way that the numbers $d$ and $d/2$ have the same meaning, for every $d\geq 0$. This is to make sure that everything that can be packed into a number (which may be very large) can actually be encoded by any searcher under any virtual vertex, regardless of the searcher's unit of distance. This ``scalability'' property also gives a searcher the ability to get arbitrarily close to its virtual vertex without losing information, by repeatedly moving halfway towards it (note that the virtual vertex is still the closest visible vertex after this move).

\paragraph{Naive code.}
To pack our data into a single real number, we use the number's binary digits. Let us restrict our attention to the real numbers in the interval $[0,1)$. Each of these numbers is identified by the fractional part of its binary representation, which is an infinite sequence of binary digits. Moreover, if we forbid binary representations ending with an infinite sequence of digits 1, the binary representation of any real number is unique.

It is straightforward to pack a finite sequence of real numbers $(a_1, a_2, \dots, a_n)$ into a real number in $[0,1)$. We first express each $a_i$ as a \emph{sign bit} $s_i$, which is $0$ if $a_i\geq 0$ and $1$ otherwise, an infinite binary \emph{mantissa} $\left(b^{(i)}_1, b^{(i)}_2, \dots\right)$, and a non-negative binary \emph{exponent} $e_i$, such that
$$a_i=(-1)^{s_i}\cdot\sum_{j=1}^\infty b^{(i)}_j\cdot 2^{e_i-j}.$$
Then we express each exponent $e_i$, which is a non-negative integer, as the infinite sequence of binary digits $\left(e^{(i)}_1, e^{(i)}_2, \dots\right)$, such that
$$e_i=\sum_{j=1}^\infty e^{(i)}_j\cdot 2^{j-1}.$$

Hence we have $n$ sign bits to pack, plus $2n$ infinite binary sequences. We also want to fulfill the scalability requirement of our code, and so we add a \emph{scale} $\lambda$, which is a non-negative integer. Our final result is the real number whose binary representation is
$$0.0^\lambda1^n0s_1s_2\dots s_nb^{(1)}_1e^{(1)}_1b^{(2)}_1e^{(2)}_1\dots$$
 $$\dots b^{(n)}_1e^{(n)}_1b^{(1)}_2e^{(1)}_2b^{(2)}_2e^{(2)}_2\dots b^{(n)}_2e^{(n)}_2\dots.$$
By $0^\lambda$ we mean a sequence of $\lambda$ digits $0$, and by $1^n$ we mean a sequence of $n$ digits $1$. It is clear that the original sequence $(a_1, a_2, \dots, a_n)$ can be reconstructed from this number, and that the number can be made arbitrarily small by increasing $\lambda$.

Since we know how to represent a snapshot by a finite array of coordinates, we can also pack it into a single real number. Then, to pack an array of $m$ snapshots, we can simply pack each snapshot separately, and then pack the resulting $m$ numbers into a single number.

\paragraph{Real random-access machines.}
Let us see how our naive encoding (and decoding) strategy could possibly be computed, and what it means to compute a real number. Of course, a traditional Turing machine with $n$ tapes containing the binary representation of every $a_i$ could compute any digit of the naive code in finite time. However, computing all of its digits requires an infinitely long computation.

To overcome this limitation of Turing machines, some models of computation that operate directly on real numbers have been introduced. These include the \emph{Blum-Shub-Smale machine}~\cite{blum}, which is a random-access machine whose registers can store arbitrary real numbers. Its computational primitives are the four basic arithmetic operations on real numbers, and it can test (and branch) if a real number is positive. Each of these operations takes one unit of time.

Depending on the application, it is also customary to extend the basic model with additional primitives, such as root extractions, trigonometric functions, etc. Of course, the extra primitives that we include should be somewhat well-behaved and intuitively computable, or else we would defeat the purpose of using these machines as models of computation. For instance, it would be reasonable to require at the very least that our unary primitives be real functions of a real variable whose set of discontinuities is nowhere dense. This would admit all the algebraic functions, the trigonometric functions, the exponential functions, the logarithms, and many more.

\paragraph{Non-computability of the naive code.}
As it turns out, our naive encoding method is not implementable on an extended Blum-Shub-Smale machine. Let us consider the simple case in which we want to pack the two numbers $a=0.b_1b_2\dots$ and $a'=0.b'_1b'_2\dots$ into the number $f(a,a')=0.110b_1b'_1b_2b'_2\dots$. Being able to compute $f(a,a')$ for every $a$ and $a'$ is equivalent to having a primitive operator $g(x)$ that interleaves the binary digits of $x$ with $0$'s (assuming that $0\leq x<1$). Indeed, $g(x)=8\cdot f(x,0)-6$ and $f(a,a')=g(a)/8+g(a')/16+3/4$.

Assume that $x\neq 0$ has a finite binary representation, i.e., $x=0.b_1b_2\dots b_m$, with $b_m=1$. Then, $g(x)=0.b_10b_20\dots 0b_m$. Now, let $x_i=x-2^{-m-i}$. Clearly, $\lim_{i\to\infty}x_i=x$. We have $x_i=0.b_1b_2\dots b_{m-1} 01^i$, and hence $g(x_i)=0.b_10b_20\dots 0b_{m-1}00(01)^i$. So,
$$\lim_{i\to\infty}g(x_i)=0.b_10b_20\dots 0b_{m-1}00\overline{01}\neq g(x).$$
Therefore, he have
$$g\left(\lim_{i\to\infty}x_i\right)=g(x)\neq\lim_{i\to\infty}g(x_i),$$
which means that $g$ is not continuous at $x$. Recall that $x$ was a generic number with a finite binary representation. Hence, $g$ is discontinuous on the set of rationals of the form $m/2^n$, with $0<m<2^n$, which is dense in $(0,1)$. So, according to our discussion on computability, $g$ is not a reasonable primitive for an extended Blum-Shub-Smale machine. It is not hard to generalize our argument to the naive encoding of more than two numbers, as well as the decoding functions.

\paragraph{Polygons with algebraic vertices.}
We now propose a more sophisticated encoding strategy, which is computable even on a basic Blum-Shub-Smale machine (i.e., the one with the four basic arithmetic operations only). A small drawback is that we can only apply this method if the vertices of the polygon $P$ have algebraic coordinates (i.e., they are \emph{algebraic points}) in some global coordinate system. (Recall that a real number is \emph{algebraic} if it is a root of a polynomial with integer coefficients.) Note that we do not require that the searchers' positions be algebraic points at any time during the execution. Their local units of distance do not have to be algebraic, either. As a consequence, even under our assumptions, the snapshots of $P$ that the searchers get do not necessarily have vertices with algebraic coordinates.

In practice, we are not imposing a big limitation on our inputs, in that basically all the polygons we can reasonably think of fall into this class. Indeed, the algebraic numbers include the rationals and are closed under basic arithmetic operations and extractions of roots of any degree~\cite{libro1}. Moreover, a simple consequence of de Moivre's formula is that the sines and cosines of all the rational multiples of $\pi$ are algebraic~\cite{lehmer}. Hence, the vertices of all the regular polygons inscribed in the unit circle and having a vertex in $(1,0)$ are algebraic points. So, for instance, we could construct the vertex set of our polygon $P$ by putting together copies of these ``unit polygons'', rotated by rational multiples of $\pi$ and scaled by rational factors. This simple scheme already yields a very rich class of polygons of all symmetricities.

\paragraph{Representing algebraic reals.}
The reason why we insist on working with algebraic numbers is that they have concise representations that can be manipulated efficiently. To understand our technique, it is worth considering the rational numbers first. The polygons with rational vertices do not constitute a very interesting class, because their symmetricity can only be $1$, $2$, or $4$ (indeed, this is equivalent to the fact that, for $n\notin\{1,2,4\}$, there are no regular $n$-gons in the plane whose vertices have integer coordinates, which in turn can be proved by standard algebraic methods~\cite{integer}). Nonetheless, discussing rational numbers allows us to expose some of the key ideas of our encoding method without getting involved with technicalities. Let $p/q$ be a rational number, with $q>0$. We can describe it by three non-negative integers: a sign bit for $p$, the absolute value $|p|$, and $q$. We represent each non-negative integer $n$ as the bit string $0^n1$, and then we simply concatenate the representations of all three numbers as the fractional part of a real number expressed in binary. For instance, the rational $5/3$ becomes $0.10000010001$ (because the sign bit of $5$ is $0$), and $-5/3$ becomes $0.010000010001$ (because the sign bit of $-5$ is $1$).

The advantage of this code over the standard binary representation is that this one is always finite. We can then retrieve the most significant bit $b_1$ of this representation by multiplying the number by $2$ and testing if the result is less than $1$. We then subtract $b_1$ from the result and we repeat the same process to retrieve $b_2$, etc. We know that all the remaining bits are $0$ when the number itself becomes $0$. With a similar technique we can modify any bit of the code, and therefore we can transform the entire code by any Turing-computable function. In particular, given the representations of two rationals $p/q$ and $p'/q'$, we can do basic computations on them without ever reconstructing the actual numbers. For instance, once we have the two pairs of integers $(p,q)$ and $(p',q')$, we can compute the sum $p/q+p'/q'$ as the pair of integers $(pq'+p'q,qq')$, without actually constructing the real number $p/q$ or the real number $p'/q'$. Note that the low-level bit manipulations that we do to achieve this are computable by a basic Blum-Shub-Smale machine.

Representing generic algebraic numbers is done in a similar way, although the procedure is complicated by some technical issues. Since the algebraic number $\alpha$ is a root of the polynomial with integer coefficients $Q(x)=a_nx^n+a_{n-1}x^{n-1}+\dots+a_1x+a_0$, we could attempt to represent it as the array of the coefficients of $Q$, i.e., $(a_n,a_{n-1},\dots,a_0)$. We may also assume that $Q$ is the \emph{minimal polynomial} of $\alpha$, which is unique. However, since $Q$ has $n$ complex roots (counted with their multiplicity), we also have to tell which of these roots we are representing. Fortunately, the real roots of $Q$ can be ordered. If $\alpha$ is the $i$th real root of $Q$, we therefore represent it by the sequence $(n,i,a_n,a_{n-1},\dots,a_0)$, which can easily be expressed as a single real number with a finite binary representation by encoding the sign bit and the absolute value of each of the integers, as we did with the rationals. Observe that we explicitly stored the number $n$ as a first thing, so we know when to stop during the decoding procedure (we may be given ``by accident'' a number with infinitely many $1$'s in its binary representation, and we do not want to get stuck in an infinite loop trying to decode it).

As we did with the rationals, once we have some algebraic numbers expressed in this finite form, we can do Turing-computable bit manipulations to compute all kinds of common functions on them. In particular, there are standard ways of computing the basic arithmetic operations, as well as root extractions of any degree. Moreover, since we are using minimal polynomials, each algebraic number has a unique code, and therefore testing if two of them are equal is trivial. A comprehensive exposition of these techniques, along with their theoretical background, is found in~\cite{libro1}. Essentially, this is also one of the several ways in which mathematical software such as Sage, Mathematica, and CGAL handles algebraic numbers and does exact computations with them.

The key point to keep in mind is that, once a number is encoded in this form, we cannot necessarily retrieve it in finite time; we can only approximate it arbitrarily well, for instance via Sturm's theorem~\cite{libro1}. However, we can still evaluate computable predicates on these numbers exactly, and have them influence the flow of our algorithms~\cite{libro1}.

\paragraph{Computable code.}
Suppose a basic Blum-Shub-Smale machine has an algebraic number $\alpha$ stored in a register; let us see how it can effectively construct its code. The machine starts generating all finite sequences of bits in lexicographic order. For each sequence, it checks if it is a well-formed code of an algebraic number; if it is, it extracts the coefficients of the polynomial $Q$ from it, as explained above. Then it computes $Q(\alpha)$, which requires only additions and multiplications of real numbers. Since $\alpha$ is algebraic, eventually a polynomial $Q$ is found such that $Q(\alpha)=0$. It is well known that $Q$ must be a multiple of the minimal polynomial of $\alpha$; hence, it is sufficient to factor $Q$ over $\mathbb Z$ and pick the irreducible factor that has $\alpha$ as a root: this factor $Q'$ must be the minimal polynomial of $\alpha$. Then Sturm's theorem can be applied to find out how many real roots of $Q'$ are smaller than $\alpha$, and this number is used along with the coefficients of $Q'$ to encode $\alpha$ (the details of this process are explained in~\cite{libro1}).

Now that we know how to compute the code of a single algebraic number, let us see how we can encode an entire snapshot of a polygon $P$ with algebraic vertices taken from a point $p\in P$ by some searcher $s$. Formally, this is the set of points of $P$ that are visible to $p$, transformed by an affine map $f_p\colon \mathbb R^2\to\mathbb R^2$. This map translates points from the global coordinate system to the coordinate system of $s$: it translates $p$ into the origin and then scales (by a non-zero factor) and rotates the plane about the origin. Note that $p$ is not necessarily an algebraic point, and the parameters of $f_p$ are not necessarily algebraic numbers. However, as $f_p$ is a similarity transformation, it preserves the ratios between segment lengths. Observe that the distance between any two vertices of $P$ is algebraic, because it is computable by basic arithmetic operations and extractions of square roots (by the Pythagorean theorem), and algebraic is therefore also the ratio between two such distances. It follows that the distances between vertices of $f_p(P)$ may not be algebraic, but all their ratios are. The same reasoning can be extended from the vertices of $P$ to all the points that are algebraic in the global coordinate system. These include the projection of any vertex of $P$ onto the line through two other vertices of $P$, because the coordinates of such a point can be computed by a rational function of the coordinates of the three vertices involved.

Now let $v$ and $v'$ be two vertices of $f_p(P)$, and let $g_{v,v'}\colon \mathbb R^2\to\mathbb R^2$ be the (unique) similarity transformation with positive scale factor that maps $v$ into $(0,0)$ and $v'$ into $(1,0)$. Based on the previous paragraph's reasoning, we can conclude that the vertices of $g_{v,v'}(f_p(P))$ are algebraic points. Indeed, let $u$ be one such vertex, let $u'=g^{-1}_{v,v'}(u)$, and let $u''$ be the projection of $u'$ onto $vv'$. We have that $f^{-1}(u'')$ is algebraic, and hence $|u.x|=\|vu''\|/\|vv'\|$ is also algebraic. Similarly, if $u'''$ is the projection of $u'$ onto the line through $v$ that is orthogonal to $vv'$ (hence $f^{-1}(u''')$ is algebraic), we have that $|u.y|=\|vu'''\|/\|vv'\|$ is algebraic, as well.

This basically means that, if $s$ picks two visible vertices of $f_p(P)$, say $v$ and $v'$, it takes the line $vv'$ as the $x$ axis and the length $\|vv'\|$ as the unit of distance, and expresses all the visible vertices of $f_p(P)$ in this new coordinate system, then these will be algebraic points, which can be encoded with our method by a basic Blum-Shub-Smale machine. The problem is that the snapshot taken from $p$ may not only contain vertices of $f_p(P)$. Recall that this snapshot is a list of sub-segments of the edges of $f_p(P)$: if an edge is only partially visible to $s$, it is seen by $s$ as a segment (or a collection of segments) with different endpoints. These endpoints may not be algebraic in the new coordinate system, and hence they cannot be encoded with our technique.

Our solution is to identify these potentially non-vertex endpoints and simply mark them with an ``undefined'' tag. These turn out to be precisely the endpoints that are not \emph{fully visible} to $s$ (cf.~Figure~\ref{f:a}). For instance, let $S=(x_1, y_1, x'_1, y'_1, x_2, y_2, x'_2, y'_2, \dots)$ be the snapshot received by $s$, and suppose that $(x_i,y_i)=c\cdot (x_j,y_j)$ for some $0<c<1$. Then, we mark $(x_j,y_j)$ with a special tag. More precisely, we add an ``undefined'' bit to all the entries of the snapshot, and we set it to $0$ or $1$, depending if the corresponding point is certainly a vertex of $f_p(P)$ or possibly not a vertex. Note that this check can be done by a basic Blum-Shub-Smale machine. Then we can pick any two ``defined'' points $v$ and $v'$ of $S$ (which obviously exist), use $vv'$ as the $x$ axis, and transform all the ``defined'' points of $S$ as detailed above. Each ``undefined'' point of $S$ is simply replaced with a $(0,0)$ (preserving its ``undefined'' tag), or any algebraic point of our choice. The result is a transformed snapshot $S'$ whose points are guaranteed to be all algebraic. Hence we can effectively encode their coordinates with a finite number of bits, and then concatenate all these sequences of bits into the binary representation of a single real number. We also encode the ``undefined'' bits and the indices of $v$ and $v'$ in $S$. Everything is preceded by the total number of elements in the code: as usual, this is to let the decoding procedure know when to stop. We denote the final result by $C(p,v,v')$.

Now, given $C(p,v,v')$, the searcher $s$ can decode it and reconstruct the ``defined'' vertices of the snapshot, as well as the edges between them. The coordinates of these vertices are still in our implicit form, but the searcher can operate on them, computing new algebraic points, again in the same implicit form. However, if $s$ is currently in $p$, and therefore has access to the original snapshot $S$, it can easily retrieve the actual coordinates of $v$ and $v'$, because their indices are stored in $C(p,v,v')$ (and they are plain integers). So, suppose that $s$ has computed a point $q$ in implicit form based on $C(p,v,v')$. By Sturm's theorem, it can explicitly construct a point $q'$ that is arbitrarily close to $q$ (i.e., $q'$ is not encoded in our implicit form, but it is a real number on which the machine can directly operate). Then, knowing the coordinates of $v$ and $v'$, $s$ can transform $q'$, via rational functions, back into the coordinate system in which $S$ is expressed (which is the local coordinate system of $s$). Knowing how close $q'$ is to $q$ (which is a parameter of Sturm's theorem that $s$ can set), and knowing the determinant of the transformation, $s$ can infer how close the resulting point is to the real one. In particular, if $q$ is supposed to represent a fully visible vertex of the polygon, $s$ can determine which vertex it is in finite time, by computing a good-enough approximation of it, and comparing it with the points in $S$.

We can pack any list of $m$ snapshots into a single real number by encoding $m$, followed by the codes of all the snapshots. Along with the snapshots, we can also pack as many other finitely described elements as we want. We may add a fixed-length ``label'' to the code of each element, describing its content and specifying if it represents a snapshot, an integer, etc. As with naive encoding, we also put a sequence of the form $0^\lambda1$ as a first thing in our code, where $\lambda$ is the scale.

Observe that our encoded snapshots are not exact copies of the real ones, because some information about the ``undefined'' points is lost. In the second part of this section, we will show that the information that we encode is enough for the purpose of our application to the Meeting problem.

\subsection{Adapting the Algorithms}\label{s:4.2}
Next we are going to apply our encoding methods to the Meeting algorithms of Section~\ref{s:3}, and we will show how oblivious searchers can solve the Meeting problem, as well. We will be using the improved encoding, so we will need searchers to be able to compute only basic arithmetic operations on real numbers, as well as extract square roots. Hence, internally, each searcher will run a Blum-Shub-Smale machine extended with a square-root primitive. Only the four basic arithmetic operations are required for our computable encoding method, but square roots are needed in the geometric computations. It is well known that the points whose coordinates can be computed by composing these five operations are precisely the ones that can be constructed with a compass and a straightedge~\cite{kaza}. (In turn, the Mohr-Mascheroni theorem states that these points can also be constructed with a compass alone~\cite{mascheroni}.)

\paragraph{Main ideas.}
Recall that both our Meeting algorithms work by making searchers jump from one vertex of the augmented polygon $\widetilde P$ to another. This behavior is roughly compatible with the idea of simulating memory by moving close enough to a vertex of $P$ and encoding information as the distance from it (in the terminology of Section~\ref{s:4.1}, this is called the \emph{virtual vertex}). When activated, a searcher will compute its distance from the virtual vertex (note that this requires the extraction of a square root), and it will decode this distance, thus retrieving its lost memory. It will then execute one of the old algorithms, ``pretending'' to be located exactly on the virtual vertex. Instead of moving onto the destination vertex, it will move close enough to it, re-encoding its entire memory plus the newest snapshot. Of course, this technique introduces several issues.
\begin{itemize}
\item Recall that some information is lost in the encoding of our snapshots, because some points are marked as ``undefined''. We have to make sure that this loss of information does not invalidate the correctness of our algorithms. (Indeed, we will show that the algorithms work as intended even if the vertices recorded in the encoding of a snapshot are just the fully visible ones.)
\item As explained later, each snapshot is encoded by first re-casting it into a different coordinate system, which is not necessarily the searcher's local one. A searcher may not be able to reconstruct this special coordinate system after it moves and its virtual vertex changes. We have to show how a searcher can ``transport'' snapshots around $P$ without compromising their usability. (The solution is to use a coordinate system where the $x$ axis is marked by the current and next virtual vertices, so the searcher can reconstruct it after moving.)
\item The EXPLORE phase of the algorithms relies on the connectedness of the visibility graph of $P$. If a searcher explores $P$ by approaching its vertices but without properly touching then, it may be unable to discover some unexplored vertices. We have to show how to avoid this situation. (This is resolved by making the searcher move close enough to all vertices of $P$ and to their angle bisectors.)
\item The tours performed in the PATROL phase turn at the pivot point and at the vertices of the augmented polygon $\widetilde P$, which are not necessarily vertices of $P$. Unfortunately, oblivious searchers cannot approach generic points without losing information. (We will show how to modify their paths to make them turn only at vertices, without compromising the correctness of the PATROL phase.)
\item During the PATROL phase, two searchers are supposed to become mutually aware, either because they travel on the same edge or diagonal of $P$ or because they reach the same triangle of a special triangulation. Once again, if searchers follow their predefined routes only approximately, they may fail to meet. (We will show how to avoid this by making the searchers move within a thin-enough band that approximates the intended path.)
\end{itemize}
In the following, we will address all these issues in greater detail.

\paragraph{Approaching vertices.}
In order to apply our encoding strategy, we must first ensure that a searcher has a well-defined virtual vertex. If a searcher has more than one closest visible vertex, it just moves to one of them. Similarly, if at any time a searcher realizes that the information it is currently encoding is either internally incoherent or contrasts with the current snapshot, it moves onto its virtual vertex. So, when a searcher finds itself on a vertex, it knows that it has to restart the execution from the beginning.

In all other cases, a searcher has a destination vertex, and it moves close enough to it. It may not be able to determine right away how close it has to move, but it can reduce this distance later, if needed. In Section~\ref{s:4.1} we introduced the \emph{scale} of a code, and we argued that a searcher can always get as close as it wants to its virtual vertex (by adjusting the scale) without losing information. This ``approaching move'' keeps the searcher on the same ray emanating from the virtual vertex, and it is therefore useful when the searcher wants to maintain a certain angle with respect to the virtual vertex's incident edges.

In general, during the EXPLORE phase, before choosing its next destination vertex, a searcher $s$ will first adjust its position around its virtual vertex $v$ in such a way that all the points of its two incident edges $vv'$ and $vv''$ become \emph{fully visible} to $s$ (cf.~Figure~\ref{f:a}). This is relatively easy to do, since $s$ has access to the current snapshot. If the interior of one of the incident edges of $v$, say $vv''$, is completely invisible to $s$, it means that $v$ is a reflex vertex of $P$, and $s$ can see at least part of the edge $vv'$. In this case, $s$ moves to the line through $v$ perpendicular to $vv'$. The destination point $p'$ is chosen in such a way that the circle through $v$ centered in $p'$ intersects the boundary of $P$ only in $v$ (see Figure~\ref{f:8}). This guarantees that $v$ will be the virtual vertex again. From there, if the interior of $vv''$ is still completely invisible, $s$ moves to the extension of the segment $vv'$, again keeping $v$ as the virtual vertex (this move is always possible, although the destination of $s$ may have to be much closer to $v$ than $p'$ is: i.e., the two circles in Figure~\ref{f:8} do not necessarily have the same radius). After this move, both incident edges of $v$ will be at least partially visible. Then, $s$ approaches $v$ until it can see $vv'$ and $vv''$ entirely.

\begin{figure}[ht]
\centering
\includegraphics[scale=1.25]{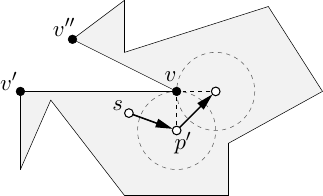}
\caption{Moving around the reflex vertex $v$ to see both its incident edges}\label{f:8}
\end{figure}

When $s$ finally sees both incident edges of $v$, it decides what its next virtual vertex $u$ will be. Let $s$ be currently located in $p$. Then, $u$ has to be a vertex of $P$ that is fully visible to every point on the segment $vp$. $s$ also has to choose a destination point $p'$, again fully visible to every point on $vp$. Moreover, among the vertices that are visible to $p'$, $u$ should be the closest to $p'$. Since by our assumption $u$ is fully visible to $p$, a suitable point $p'$ can always be found by $s$. Namely, if $u$ is a convex vertex of $P$, then $p'$ will be chosen close enough to $u$ on its angle bisector (an entire neighborhood of $u$ is visible to $p$, so this is easy to do). If $u$ is a reflex vertex, then $p'$ will be the center of a circle that touches the boundary of $P$ only in $u$.

\paragraph{Transporting snapshots.}
Recall that snapshots are encoded in a coordinate system defined by two vertices of $P$, which guarantees that the vertices in the snapshot can be encoded as algebraic points (provided that the vertices of $P$ are algebraic in some global coordinate system to begin with). Using the notation introduced in Section~\ref{s:4.1}, we will assume that all the $n$ snapshots that searcher $s$ is currently encoding are of the form $C(p_i,v,v')$, with $1\leq i\leq n$. In our notation, $p_i$ is the point from which the $i$th snapshot was taken, $v$ is the current virtual vertex of $s$, and $v'$ is another vertex of $P$ that is fully visible to all the points in the segment $vp$, where $p=p_{n+1}$ is the current location of $s$. So, all the snapshots that $s$ ``remembers'' are encoded in the same coordinate system, defined by $v$ and $v'$. Along with the snapshots, $s$ also remembers a rational approximation of $v'-v$, expressed in the local coordinate system of $s$. This approximation is assumed to be so good that $s$ can retrieve the coordinates of $v'$ (in its local coordinate system) by looking at its current snapshot. Knowing the coordinates of $v$ and $v'$, $s$ can then re-map every $C(p_i,v,v')$ into its local coordinate system, and compute arbitrarily good approximations of any algebraic point that it constructs implicitly.

Of course, as $s$ moves around $v$ and approaches it as explained before, it must always make sure that, whenever it moves from $p$ to $p'$, every point of $vp'$ is fully visible to $v'$. This is done by choosing $p'$ close enough to $v$, and it is possible because $v'$ is fully visible to all points of $vp$, by our assumption.

Suppose now that $s$, currently located in $p=p_{n+1}$, intends to change virtual vertex from $v$ to $u$. By our assumption, it does so only if $u$ is fully visible to all points of $vp$. In order to preserve the readability of the snapshots that $s$ is encoding, it has to convert them from the coordinate system in which they are expressed into a different one, which will allow $s$ to reconstruct the snapshots from a neighborhood of $u$. To do so, $s$ converts each $C(p_i,v,v')$ into $C(p_i,u,v)$. Since $s$ knows the exact positions of the three vertices involved (i.e., $u$, $v$, $v'$), because they are all in the current snapshot, it can perform this conversion, which is simply a change of coordinates computable with a rational function. Then $s$ encodes the current snapshot in the same coordinate system as the others, obtaining $C(p_{n+1},u,v)$. Finally, $s$ computes $(v+p)/2-u$ and suitably truncates the binary representation of its coordinates, obtaining a finite approximation $w$ of it. The approximation must be good enough, so that the point $u+w$ is in the interior of $P$ and closer to $v$ than any vertex of $P$ that is currently visible to $s$ (apart from $v$ itself). Then, $s$ computes a destination point $p'$ whose distance to $u$ encodes $w$ (whose binary representation is finite), followed by the snapshots $C(p_i,u,v)$, with $1\leq i\leq n+1$ (note that constructing $p'$ involves a square root extraction). When $s$ gets to $p'$, it finds its virtual vertex $u$ and retrieves $w$ (as an explicit rational vector) from $\|up'\|$. Then it computes $u+w$, which is an approximation of the midpoint of the segment $vp$. Since by our assumption all points of $vp$ are fully visible to $p'$, it is easy for $s$ to identify $v$. Now $s$ can retrieve the snapshots $C(p_i,u,v)$ from $\|up'\|$ and use $u$ and $v$ to re-map them into its own coordinate system.

\paragraph{Exploring the polygon.}
Suppose a $P$-searcher $s$ has successfully decoded its distance from its virtual vertex $v$, obtaining a history of snapshots. Since all these snapshots are expressed in the same coordinate system, it is trivial for $s$ to merge them all together and check if the common parts of two different snapshots match. If they do not match, it means that the current position of $s$ does not encode anything meaningful, and so $s$ moves to $v$. We have already explained how $s$ can reconstruct the coordinate system of these snapshots, and use it to encode the new snapshot taken from its current location in the same fashion. When all these snapshots (including the current one) have been tested against each other and merged, the result is a self-consistent \emph{collective snapshot} $S$, which is supposed to represent the part of $P$ that $s$ has already visited.

Also, whenever $s$ encodes its current snapshot, it marks the position of its virtual vertex with a special ``visited'' flag. This flag is preserved when a snapshot is transported and converted to a different coordinate system. So, when $s$ constructs the collective snapshot $S$, it also knows what vertices of $S$ have already been visited.

The EXPLORE phase begins with $s$ on a vertex $v$, and the strategy is to keep following the connected component $C$ of the boundary of $P$ that contains $v$, always in the same direction (either clockwise or counterclockwise), while encoding all the snapshots taken. This is easy to do, because we have explained how $s$ can adjust its position around its virtual vertex so that both its incident edges become fully visible.

Upon completing its first tour of $C$, $s$ has a full picture of it in the collective snapshot $S$, and starts a second tour of $C$ in the same direction, this time carefully choosing its destination points, as explained next. Let $v=v_1$, $v_2$, \dots, $v_m$ be the vertices of $C$, in the order $s$ is following them. In the second tour, for each $v_i$, $s$ wants to reach a point $p_i$ close enough to it, so that the polygonal chain $\overline C=(p_1,p_2,\dots, p_m)$ does not self-intersect (i.e., it is the boundary of a simple polygon), and does not intersect the boundary of $P$. For instance, $p_i$ may be chosen on the angle bisector of $v_i$ and close enough to it. So, upon reaching the angle bisector of $v_i$ (during the second tour), $s$ uses the information in $S$ to compute how close to $v_i$ it has to get to construct a suitable $p_i$. An adequate distance $d$ is computed implicitly, and then $s$ can get an approximation of it in explicit form and choose a distance that is certainly smaller than $d$.

When the second tour is complete and $s$ has touched all vertices of $\overline C$, it picks the first vertex $v'$ of $P$ that is in $S$ but is not yet marked as visited. Then $s$ follows a shortest path to $v'$ in which each vertex touched is fully visible to the previous one. Note that this path obviously exists, because if a vertex in the path does not fully see the next one, it can preliminarily move toward the closest vertex that is on the same line and is obstructing its vision. Once $v'$ has become the virtual vertex, $s$ follows the same exploration procedure on the connected component of the boundary of $P$ that contains $v'$, say $C'$, which is necessarily disjoint from $C$. During the second tour of $C'$, $s$ traces an approximated polygonal chain $\overline C'$ as before, but with the additional requirement that it does not intersect $\overline C$. This can be done in the same way as with $\overline C$, by computing a thin-enough ``band'' around $C'$ and making sure to move within it.

This general procedure is repeated as long as new vertices of $P$ are discovered and appear in $S$ as unvisited. Each time a new connected component $C_j$ of the boundary is discovered, $s$ follows it and constructs an approximation $\overline C_j$ that does not intersect any of the previously constructed ones and is also disjoint from the boundary of $P$. So, when the procedure ends, $s$ has touched the vertices of some mutually disjoint simple closed polygonal chains $\overline C_j$, with $1\leq j\leq m$, and each vertex of $P$ is either undiscovered or marked as visited. To prove the correctness of the EXPLORE phase, we have to show that in this situation all vertices of $P$ have indeed been discovered. Let us construct a new polygon $P'$ by removing every $C_j$ from the boundary of $P$ and replacing it with $\overline C_j$. $P'$ is indeed a polygon because of the way the $\overline C_j$'s have been constructed. Also, if two points fully see each other in $P'$, they must also fully see each other in $P$. Moreover, the unvisited vertices of $P$ are also vertices of $P'$. If some unvisited vertices exist, then one of them, say $u$, must be fully visible to a vertex $u'$ of some $C_j$, because $P'$ is connected. Since $u'$ is a vertex of $C_j$, $s$ must have been exactly in $u'$ and must have taken a snapshot from there. Recall that fully visible vertices are never marked as ``undefined'' in the encoded snapshots, and so $s$ must have carried around the implicit coordinates of $u$, which therefore must appear in the collective snapshot $S$. This is a contradiction, and therefore our exploration procedure is correct.

\paragraph{Basic patrolling.}
Let us show how to adapt the PATROL phase of Algorithm~\ref{f:2} to oblivious searchers. Each $P$-searcher that executes the EXPLORE phase correctly ends up with a collective snapshot $S$ that is a faithful copy of $P$ expressed in implicit form in a different coordinate system. So, all the similarity-invariant geometric constructions made by our old algorithm can be performed again by oblivious searchers. The first technical issue here is that some of the points generated by these constructions are not vertices of $P$ but midpoints of edges. Since we typically want the destination point of a searcher to be as close as possible to a vertex, we cannot make searchers turn at midpoints of edges as they patrol the boundary of $P$, but only turn at vertices.

Let us see how we can modify the tour of the boundary of $P$ so that it only turns at vertices, without invalidating the correctness of the algorithm. Patrolling the boundary of $P$ is what makes the map construction algorithm self-stabilizing: as a searcher repeatedly touches every vertex of $P$, it is able to tell if its memory contents are incorrect. So, we do not necessarily have to augment $P$ exactly as we did in Section~\ref{s:3.1}; we only have to make each searcher follow a path that touches every vertex. Of course, we also want two searchers with the same pivot point to follow the same path, so they necessarily meet while patrolling it. If $P$ is not axially symmetric, this is simple, and Algorithm~\ref{f:2} already does it. If $P$ is axially symmetric, we just make each searcher's path symmetric with respect to the axis that passes through its pivot point, in any (deterministic) way that we want. This is to make sure that if two searchers have that pivot point, they will compute the same path even if they have a different notion of clockwise direction. Figure~\ref{f:6} shows an example of how such a path may be constructed in an axially symmetric polygon. Observe that the pivot point in this example is in the middle of an edge, and the tour covers that edge twice. This makes the last vertex of the clockwise tour coincide with the first vertex of the counterclockwise tour, and vice versa.

\begin{figure}[ht]
\centering
\includegraphics[scale=0.9]{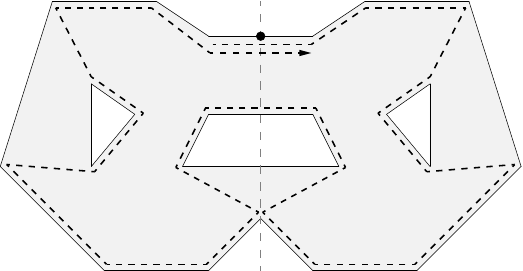}
\caption{Axially symmetric tour that visits all vertices and only turns at vertices}\label{f:6}
\end{figure}

There are two remaining issues. First, a searcher performing a tour of the boundary of $P$ without actually touching its vertices may be unable to ever detect if the collective snapshot that it is encoding has any mistakes. Second, two searchers may fail to meet even if they have the same pivot point and perform the same tour in opposite directions, because they only trace an approximation of that tour. Fortunately, both these issues have the same solution, which is the one we already described for the EXPLORE phase. Namely, as a searcher already has a full picture of $P$ (or what it ``believes'' to be $P$) it can pre-compute a thin-enough ``band'' around the path it intends to follow, and always move within this band. This can be done even if the band is only computed implicitly, as we showed for the EXPLORE phase. Also, whenever a searcher reaches a new virtual vertex, it makes sure to stop on its angle bisector and take a snapshot from there. This way, if $P$ and the collective snapshot of the searcher have discrepancies, the searcher will eventually see a missing vertex, a misplaced vertex, or an extra vertex, and this is proven exactly as we did for the EXPLORE phase. Moreover, as both searchers remain within the same thin band, they must become mutually aware as soon as they cross each other on the same edge or around the same vertex. Note that computing a band with these properties is a straightforward geometric problem that can be solved locally by the searchers given a representation of $P$.

The rest of the proof of correctness is the same as for Theorem~\ref{t:2}, which is thus extended to oblivious searchers.

\begin{theorem}
There is an algorithm that, for every integer $\sigma>0$, solves the Meeting problem with $\sigma+1$ oblivious searchers in every polygon with symmetricity $\sigma$. If the polygon's vertices are algebraic points, the algorithm is implementable on a real random-access machine that can compute basic arithmetic operations and extract square roots.\qed
\end{theorem}

\paragraph{Improved patrolling.}
We can adapt the improved patrolling strategy of Algorithm~\ref{f:4} to oblivious searchers almost in the same way as we did with the basic one. However, the $j$-tours and their ``hierarchy'' must be defined carefully. The problem is once again that we have to decide what to do with the triangles of the triangulation of $\widetilde P$ whose vertices are midpoints of edges or of diagonals of $P$ (cf.~Figure~\ref{f:5}). If $P$ is not axially symmetric, this is not a real problem: we can augment $P$ by cutting it along some diagonals, never creating those improper vertices. The polygons $P_j$ and the $j$-tours are then defined in the same way as in Section~\ref{s:3.2}.

Let us now focus on the case in which $P$ is axially symmetric, and let a \emph{branch} be a connected component of $P\setminus Q$, where $Q$ is the central polygon, as defined in Section~\ref{s:3.2}. In this case, we want our $j$-tours to be axially symmetric, as well. Our solution is to preliminarily construct an axially symmetric partition of each branch of $P$ into triangles and isosceles trapezoids (a trapezoid is \emph{isosceles} if its base angles are the same). This is done in a similarity-invariant way by drawing diagonals between vertices of $P$ that can fully see each other, as shown in Figure~\ref{f:7}. Note that this is made possible by the presence of isosceles trapezoids, because a symmetric branch may not have a symmetric triangulation.

\begin{figure*}[ht]
\centering
\includegraphics[scale=1]{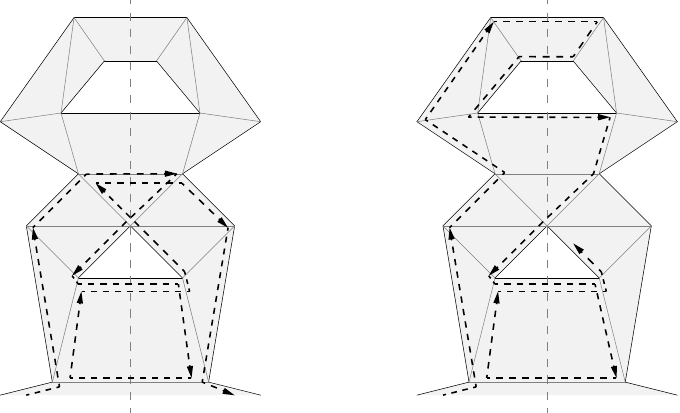}
\caption{Symmetric partition of a branch with a $5$-tour and part of a perimeter tour}\label{f:7}
\end{figure*}

Let us cut each branch along its axis of symmetry, obtaining two \emph{sub-branches}. Let us add cuts along symmetric edges of the partition constructed before until we obtain a simply connected augmented polygon $\widetilde P$. As in Section~\ref{s:3.2}, the dual graph of the resulting partition of $\widetilde P$ is a tree rooted in the central polygon $Q$. We then define a \emph{$j$-tour} as the tour of the pieces of the partition whose depth in the tree is at most $j$; such a tour must never cross the edges of $\widetilde P$, with one exception. Let $\ell$ be the axis of symmetry of a branch. When the tour is in that branch and it is supposed to follow a sub-segment of $\ell$ that splits a triangle or a trapezoid in two symmetric parts, the $j$-tour includes both parts. This construction is illustrated in Figure~\ref{f:7}. Our new $j$-tours have the required axes of symmetry and only turn at vertices of $P$. Note that a $j$-tour may now self-overlap and touch some vertices of $P$ multiple times, but this is not going to be a problem.

These new $j$-tours have the relevant property that was required by Lemma~\ref{l:2}: if a piece $T$ of the partition is included in the $(j+1)$-tour but not in the $j$-tour, then an edge $e$ of $T$ is part of the $j$-tour. Since $T$ is either a triangle or a trapezoid, it is convex. Therefore, a searcher performing a $j$-tour completely sees $T$ when it touches $e$. Due to this property, the proof of Lemma~\ref{l:2} goes through even for our new $j$-tours. The same holds for Lemma~\ref{l:1}, which only requires the convexity of the pieces.

It remains to explain how the new $j$-tours can be approximated by oblivious searchers without losing the aforementioned properties. We view a $j$-tour as a closed polygonal chain enclosing some pieces of the partition. The key idea is to make a searcher perform a $j$-tour by following this polygonal chain without ever properly crossing it. This way, the searcher never enters pieces of the partition where it is not supposed to go, yet. Also, when the $j$-tour covers an edge $e$ as defined above, the searcher makes sure to effectively touch $e$ (without crossing it), so to see any searcher that may be in $T$. Note that the searcher cannot explicitly compute a point on $e$ based on its implicit collective snapshot. However, by approaching one endpoint of $e$, it eventually gets to see the other one, as well (because the endpoints fully see each other). When both endpoints are visible, and hence readable in explicit form from the current snapshot, a precise move on $e$ is possible.

Other than this, $j$-tours are approximated as with basic patrolling, i.e., remaining within implicitly defined thin-enough bands around them. All these features combined enforce the properties that make Lemmas~\ref{l:1} and~\ref{l:2} valid for approximated $j$-tours. We still have to guarantee that any discrepancy between $P$ and the collective snapshot of a searcher $s$ will be detected during a perimeter tour. This is done as with basic patrolling, by making $s$ stop on the angle bisector of each vertex $v$ (and close enough to $v$). This may not be possible right away, because when $s$ reaches $v$ for the first time it may be forced to remain within a piece of the partition that does not contain the angle bisector of $v$. However, since the perimeter tour encloses all pieces of the partition, eventually $s$ will reach $v$ again and will be allowed to stop on the angle bisector.

This concludes the proof that Theorem~\ref{t:3} can be extended to oblivious searchers.

\begin{theorem}
There is an algorithm that solves the Meeting problem with two oblivious searchers in every polygon whose barycenter does not lie in a hole. If the polygon's vertices are algebraic points, the algorithm is implementable on a real random-access machine that can compute basic arithmetic operations and extract square roots.\qed
\end{theorem}

\section{Conclusions and Further Work}\label{s:5}
\paragraph{Summary.}
We have minimized the number of searchers that are required to solve the Meeting problem in an unknown polygon as a function of its symmetricity. Additionally, we showed that two searchers are sufficient in all but a small class of polygons (namely, the rotationally symmetric ones with center in a hole). We have done so even if the searchers are anonymous, asynchronous, memoryless, and can be initially located anywhere in the polygon. Moreover, if the vertices of the polygon are algebraic points in a global coordinate system, the searchers only have to compute basic arithmetic operations and square roots. As a main tool, we have used a self-stabilizing map construction algorithm of independent interest.

\paragraph{Termination detection.}
An interesting question is wheth\-er two searchers can realize when they have become mutually aware and actually \emph{terminate} the execution of the Meeting algorithm. This is not a trivial problem, because searchers are asynchronous: a searcher $s_1$ that sees another searcher $s_2$ cannot in general be sure that $s_2$ is not going to disappear behind a corner, because $s_2$ may currently be in the middle of a movement. So, $s_1$ may have to wait indefinitely to find out. However, we can show that termination is possible if the searchers execute a slightly modified version of Algorithm~\ref{f:2}, provided that the polygon has no holes. On the other hand, the termination problem remains open with regard to Algorithm~\ref{f:4} and polygons with holes.

\paragraph{Limited visibility.}
We may wonder if the Meeting problem can be solved if searchers have \emph{limited visibility}, i.e., they can only see up to a fixed distance, which is the same for all searchers. If the searchers have memory, we can adapt our algorithms of Section~\ref{s:3} by making each searcher take small-enough steps and also explore the interior of the polygon, as opposed to just its boundary, in order to detect hidden holes. The improved PATROL phase works by splitting each branch into thin-enough sub-branches and then finely triangulating each sub-branch. If searchers are memoryless, our algorithm with basic patrolling can also be adapted, provided that the polygon has no holes, and that data is encoded as the distance from the boundary of the polygon (as opposed to the distance from the closest vertex). In all other cases, the Meeting problem for memoryless searchers with limited visibility is open.

\paragraph{Non-rigid movements.}
Our algorithms for searchers with memory also work in the \emph{non-rigid} setting, i.e., when a searcher can be stopped by the scheduler during each Move phase before reaching its destination point, but not before having moved by at least a constant $\delta$ (for details on this model, refer to~\cite{librosantoro}). However, since our oblivious searchers have to make precise movements to implicitly encode memory, we cannot extend our memoryless algorithms to this model.

\paragraph{Optimizing movements.}
An interesting optimization problem is to improve our algorithms so that the total distance traveled by the searchers or the number of steps they take is minimized. In the EXPLORE phase we could visit the visibility graph of the polygon in depth-first order, which would yield a linear number of steps with respect to the number of vertices of the polygon. Our basic PATROL phase is worst-case optimal, because the searchers must visit the entire boundary of the polygon (due to their possibly incorrect initial memory states), and they indeed meet after a constant number of tours. Our improved PATROL phase could be optimized, because we chose to perform many more perimeter stages than needed. In fact, we could reduce this number from quadratic to linear in the number of vertices of the polygon.

\medskip
\paragraph{Acknowledgments.}
The authors wish to thank Francesco Veneziano for a clarifying discussion. This research has been supported in part by the Natural Sciences and Engineering Research Council of Canada under the Discovery Grant program and by Prof.~Flocchini's University Research Chair.


\begin{thebibliography}{99}

\bibitem{libro2}
S. Alpern and S. Gal.
\newblock {\em The theory of search games and rendezvous.}
\newblock Springer, 2003.

\bibitem{limited1}
H. Ando, Y. Oasa, I. Suzuki, and M. Yamashita.
\newblock Distributed memoryless point convergence algorithm for mobile robots with limited visibility.
\newblock {\em IEEE Transactions on Robotics and Automation}, 15(5):818--828, 1999.

\bibitem{blum}
L. Blum, F. Cucker, M. Shub, and S. Smale.
\newblock {\em Complexity and real computation.}
\newblock Springer-Verlag New York, 1998.

\bibitem{approach}
S. Bouchard, M. Bournat, Y. Dieudonn\'e, S. Dubois, and F. Petit.
\newblock Asynchronous approach in the plane: A deterministic polynomial algorithm
\newblock in {\em Proceedings of the 31st International Symposium on Distributed Computing (DISC)}, 8:1--8:16, 2017.

\bibitem{CDDMW13b} 
J. Chalopin, S. Das, Y. Disser, M. Mihal\'ak, and P. Widmayer.
\newblock Mapping simple polygons: how robots benefit from looking back.
\newblock {\em Algorithmica} 65(1):43--59, 2013.

\bibitem{CDDMW15}
J. Chalopin, S. Das, Y. Disser, M. Mihal\'ak, and P. Widmayer.
\newblock Mapping simple polygons: The power of telling convex from reflex.
\newblock {\em ACM Transactions on Algorithms} 11(4):33:1--33:16, 2015.

\bibitem{item1} 
M. Cieliebak, P. Flocchini, G. Prencipe, and N. Santoro.
\newblock Distributed computing by mobile robots: Gathering.
\newblock {\em SIAM Journal on Computing}, 41(4):829--879, 2012.

\bibitem{libro1}
H. Cohen.
\newblock {\em A course in computational algebraic number theory.}
\newblock Springer, 1993.

\bibitem{pelc2} 
J. Czyzowicz, D. Ilcinkas, A. Labourel, and A. Pelc.
\newblock Asynchronous deterministic rendezvous in bounded terrains.
\newblock {\em Theoretical Computer Science}, 412(50):6926--6937, 2011.

\bibitem{pelc3} 
J. Czyzowicz, A. Kosowski, and A. Pelc.
\newblock Deterministic rendezvous of asynchronous bounded-memory agents in polygonal terrains.
\newblock {\em Theory of Computing Systems}, 52(2):179--199, 2013.

\bibitem{pelc1} 
J. Czyzowicz, A. Labourel, and A. Pelc.
\newblock How to meet asynchronously (almost) everywhere.
\newblock {\em ACM Transactions on Algorithms}, 8(4):37:1--37:14, 2012.

\bibitem{rendezvous1} 
X. D\'{e}fago, M. Gradinariu, S. Messika, and P. Ra\"{i}pin-Parv\'{e}dy.
\newblock Fault-tolerant and self-stabilizing mobile robots gathering.
\newblock In {\em Proceedings of the 20th International Symposium on Distributed Computing, (DISC)}, 46--60, 2006.

\bibitem{pelc5} 
Y. Dieudonn\'{e} and A. Pelc.
\newblock Deterministic polynomial approach in the plane.
\newblock {\em Distributed Computing}, 28(2):111--129, 2015.

\bibitem{pelc4} 
Y. Dieudonn\'{e}, A. Pelc, and V. Villain.
\newblock How to meet asynchronously at polynomial cost.
\newblock {\em SIAM Journal on Computing}, 44(3):844--867, 2015.

\bibitem{self2}
Y. Dieudonn\'{e} and F. Petit.
\newblock Self-stabilizing gathering with strong multiplicity detection.
\newblock {\em Theoretical Computer Science}, 428:47--57, 2012.

\bibitem{angles}
Y. Disser, M. Mihal\'ak, and P. Widmayer.
\newblock Mapping polygons with agents that measure angles.
\newblock In {\em Algorithmic Foundations of Robotics X}, 415--425, 2013.

\bibitem{DFGPSV}
G.\,A. Di Luna, P. Flocchini, S. Gan Chaudhuri, F. Poloni, N. Santoro, and G. Viglietta.
\newblock Mutual visibility by luminous robots without collisions.
\newblock {\em Information and Computation}, 254(3):392--418, 2017.

\bibitem{DISC}
G.\,A. Di Luna, P. Flocchini, N. Santoro, G. Viglietta, and M. Yamashita.
\newblock Meeting in a polygon by anonymous oblivious robots
\newblock in {\em Proceedings of the 31st International Symposium on Distributed Computing (DISC)}, 14:1--14:15, 2017.

\bibitem{librosantoro}
P. Flocchini, G. Prencipe, and N. Santoro.
\newblock {\em Distributed computing by oblivious mobile robots.}
\newblock Morgan \& Claypool, 2012.

\bibitem{circle}
P. Flocchini, G. Prencipe, N. Santoro, and G. Viglietta.
\newblock Distributed computing by mobile robots: Uniform Circle Formation.
\newblock {\em Distributed Computing}, 30(6):413--457, 2017.

\bibitem{gatheringvisibiliy}
P. Flocchini, G. Prencipe, N. Santoro, and P. Widmayer.
\newblock Gathering of asynchronous robots with limited visibility.
\newblock {\em Theoretical Computer Science}, 337(1--3):147--168, 2005.

\bibitem{rendezvousviglietta}
P. Flocchini, N. Santoro, G. Viglietta, and M. Yamashita.
\newblock Rendezvous with constant memory.
\newblock {\em Theoretical Computer Science}, 621:57--72, 2016.

\bibitem{mascheroni}
N. Hungerb\"uhler.
\newblock A short elementary proof of the Mohr-Mascheroni theorem.
\newblock {\em The American Mathematical Monthly}, 101(8):784--787, 1994.

\bibitem{kaza}
N.\,D. Kazarinoff.
\newblock {\em Ruler and the round: classic problems in geometric constructions.}
\newblock Dover, 2003.

\bibitem{geometry}
A.\,P. Kiselev.
\newblock {\em Kiselev's geometry book I. Planimetry.}
\newblock Sumizdat, 2006.

\bibitem{integer}
D. Klobu\v{c}ar.
\newblock On nonexistence of an integer regular polygon.
\newblock {\em Mathematical Communications}, 3(1):75--80, 1998.

\bibitem{lehmer}
D.\,H. Lehmer.
\newblock A note on trigonometric algebraic numbers.
\newblock {\em The American Mathematical Monthly}, 40(3):165--166, 1933.

\bibitem{nearg}
L. Pagli, G. Prencipe, and G. Viglietta.
\newblock Getting close without touching: Near-Gathering for autonomous mobile robots.
\newblock {\em Distributed Computing}, 28(5):333--349, 2015.

\bibitem{S2} 
W. Rabinovich, J. Murphy, M. Suite, M. Ferraro, R. Mahon, P. Goetz, K. Hacker, W. Freeman, E. Saint Georges, S. Uecke, and J. Sender.
\newblock Free-space optical data link to a small robot using modulating retroreflectors.
\newblock In {\em Proceedings of SPIE 7464, Free-Space Laser Communications IX}, 7464, 2009.

\bibitem{C1}
I. Rekleitis, V. Lee-Shue, A. Peng New, and H. Choset.
\newblock Limited communication, multi-robot team based coverage.
\newblock In {\em Proceedings of the IEEE International Conference on Robotics and Automation (ICRA)}, 3462--3468, 2004.

\bibitem{SBM15}
G. Sharma, C. Busch, and S. Mukhopadhyay.
\newblock Mutual visibility with an optimal number of colors.
\newblock In {\em Proceedings of the 11th International Symposium on Algorithms and Experiments for Wireless Sensor Networks (ALGOSENSORS)}, 196--210, 2016.

\bibitem{SVTBR16} 
G. Sharma, R. Vaidyanathan, J. L. Trahan, C. Busch, and S. Rai.
\newblock Complete visibility for robots with lights in O(1) time.
\newblock In {\em Proceedings of the 18th International Symposium on Stabilization, Safety, and Security of Distributed Systems (SSS)}, 327--345, 2016.

\bibitem{hidingpeople} 
T. Shermer.
\newblock Hiding people in polygons.
\newblock {\em Computing}, 42(2):109--131, 1989.

\bibitem{yamashita2} 
I. Suzuki and M. Yamashita.
\newblock Searching for a mobile intruder in a polygonal region.
\newblock {\em SIAM Journal on Computing}, 21(5):863--888, 1992.

\bibitem{yamashita1} 
M. Yamashita, H. Umemoto, I. Suzuki, and T. Kameda.
\newblock Searching for mobile intruders in a polygonal region by a group of mobile searchers.
\newblock {\em Algorithmica}, 31(2):208--236, 2001.

\end{thebibliography}
\end{document}